\documentclass[a4paper,UKenglish,cleveref,autoref,thm-restate]{lipics-v2021}

\pdfoutput=1 %
\hideLIPIcs  %

\title{%
	Equitable Connected Partition and Structural Parameters Revisited: %
	$N$-fold Beats Lenstra%
}

\titlerunning{Equitable Connected Partition and Structural Parameters Revisited} %

\author{Václav Blažej}
{University of Warwick, United Kingdom \and \url{https://blazeva1.pages.fit/}}
{vaclav.blazej@warwick.ac.uk}
{https://orcid.org/0000-0001-9165-6280}
{}

\author{Dušan Knop}
{Czech Technical University in Prague, Czech Republic \and \url{https://knopdusa.pages.fit/}}
{dusan.knop@fit.cvut.cz}
{https://orcid.org/0000-0003-2588-5709}
{}

\author{Jan Pokorný}
{Czech Technical University in Prague, Czech Republic\and \url{https://pokorj54.pages.fit/}}
{pokorj54@fit.cvut.cz}
{https://orcid.org/0000-0003-3164-0791}
{}

\author{Šimon Schierreich}
{Czech Technical University in Prague, Czech Republic \and \url{https://pages.fit.cvut.cz/schiesim/}}
{schiesim@fit.cvut.cz}
{https://orcid.org/0000-0001-8901-1942}
{}

\authorrunning{V. Blažej, D. Knop, J. Pokorný, and Š. Schierreich}

\Copyright{Václav Blažej, Dušan Knop, Jan Pokorný, and Šimon Schierreich}



\category{} %

\relatedversion{} %

\nolinenumbers %

\usepackage{amssymb}
\usepackage{amsmath}
\usepackage{xspace}
\usepackage{todonotes}
\usepackage{tikz}
\usetikzlibrary{hobby,backgrounds,positioning,calc,trees,decorations.pathreplacing,calligraphy}
\usetikzlibrary{decorations.pathmorphing}
\tikzset{snake it/.style={decorate, decoration=snake}}

\usepackage{framed}
\usepackage{tabularx}

\newlength{\RoundedBoxWidth}
\newsavebox{\GrayRoundedBox}
\newenvironment{GrayBox}[1]%
   {\setlength{\RoundedBoxWidth}{.93\columnwidth}%
    \def\boxheading{#1}%
    \begin{lrbox}{\GrayRoundedBox}%
       \begin{minipage}{\RoundedBoxWidth}}%
   {   \end{minipage}
    \end{lrbox}
    \begin{center}
    \begin{tikzpicture}%
       \node(Text)[draw=black!20,fill=white,rounded corners,inner sep=2ex,text width=\RoundedBoxWidth]
             {\usebox{\GrayRoundedBox}};
        \coordinate(x) at (current bounding box.north west);
        \node [draw=white,rectangle,inner sep=3pt,anchor=north west,fill=white]
        at ($(x)+(6pt,.75em)$) {\boxheading};
    \end{tikzpicture}
    \end{center}}

\newenvironment{defproblemx}[1]{\noindent\ignorespaces%
                                \FrameSep=6pt%
                                \parindent=0pt%
			\begin{center}%
                \begin{GrayBox}{#1}%
                \begin{tabular*}{\columnwidth}{!{\extracolsep{\fill}}@{\hspace{.1em}} >{\itshape} p{0.1\columnwidth} p{0.85\columnwidth} @{}}%
            }{
                 \vspace*{-1em}
                 \end{tabular*}%
                \end{GrayBox}%
            \end{center}%
                \ignorespacesafterend
            }

\newcommand{\defProblemQuestion}[3]{%
  \begin{defproblemx}{#1}
    Input: & #2 \\
    Question: & #3
  \end{defproblemx}
}

\newcommand{\N}{\ensuremath{\mathbb{N}}}

\newcommand{\Z}{\ensuremath{\mathbb{Z}}}

\newcommand{\bigoh}{\ensuremath{\mathcal{O}}}
\newcommand{\Oh}[1]{\ensuremath{{\bigoh\left(#1\right)}}}

\newcommand{\cc}[1]{{\mbox{\textnormal{\textsf{#1}}}}\xspace}  %

\newcommand{\FPT}{\cc{FPT}}
\newcommand{\XP}{\cc{XP}}
\newcommand{\NP}{{\cc{NP}}}

\newcommand{\NPh}{\NP-hard\xspace}
\newcommand{\NPhness}{\NP-hardness\xspace}
\newcommand{\NPc}{\NP-complete\xspace}

\newcommand{\paraNP}{\cc{para-NP}}

\newcommand{\paraNPh}{\cc{para-NP}-hard\xspace}

\newcommand{\Wh}[1][1]{\cc{W[#1]}-hard\xspace}
\newcommand{\Whness}[1][1]{\cc{W[#1]}-hardness\xspace}

\newcommand{\YES}{\emph{yes}}

\newcommand{\Yes}{\YES}

\newcommand{\YesI}{{\YES-instance}\xspace}

\newcommand{\ECP}{\textsc{Equitable Connected Partition}\xspace}
\newcommand{\ECPshort}{\normalfont\textsc{ECP}\xspace}

\newcommand{\MCC}{\normalfont\textsc{Multicoloured Clique}\xspace}
\newcommand{\UBP}{\normalfont\textsc{Unary Bin Packing}\xspace}

\newcommand{\DP}{\operatorname{DP}}

\newcommand{\tw}{\ensuremath{\operatorname{tw}}\xspace} %
\newcommand{\pw}{\ensuremath{\operatorname{pw}}\xspace} %
\newcommand{\vi}{\ensuremath{\operatorname{vi}}\xspace} %
\newcommand{\nd}{\ensuremath{\operatorname{nd}}\xspace} %
\newcommand{\mw}{\ensuremath{\operatorname{mw}}\xspace} %
\newcommand{\pvcn}[1][d]{\ensuremath{\text{#1-}\hspace{-0.1em}\operatorname{pvcn}}\xspace} %
\newcommand{\td}{\ensuremath{\operatorname{td}}\xspace} %
\newcommand{\fes}{\ensuremath{\operatorname{fes}}\xspace} %

\newcommand{\pieces}{\ensuremath{\mathcal{P}}}
\newcommand{\config}{\ensuremath{\mathcal{C}}}

\Crefname{claim}{Claim}{Claims}

\EventEditors{John Q. Open and Joan R. Access}
\EventNoEds{2}
\EventLongTitle{42nd Conference on Very Important Topics (CVIT 2016)}
\EventShortTitle{CVIT 2016}
\EventAcronym{CVIT}
\EventYear{2016}
\EventDate{December 24--27, 2016}
\EventLocation{Little Whinging, United Kingdom}
\EventLogo{}
\SeriesVolume{42}
\ArticleNo{23}

\begin{document}
	
\maketitle

\begin{abstract}
	We study the \ECP (\ECPshort for short) problem, where we are given a graph $G=(V,E)$ together with an integer $p\in\N$, and our goal is to find a partition of~$V$ into~$p$~parts such that each part induces a connected sub-graph of $G$ and the size of each two parts differs by at most~$1$. On the one hand, the problem is known to be \NPh in general and \Wh with respect to the path-width, the feedback-vertex set, and the number of parts~$p$ combined. On the other hand, fixed-parameter algorithms are known for parameters the vertex-integrity and the max leaf number.
	
	As our main contribution, we resolve a long-standing open question [Enciso et al.; IWPEC~'09] regarding the parameterisation by the tree-depth of the underlying graph. In particular, we show that \ECPshort is \Wh with respect to the $4$-path vertex cover number, which is an even more restrictive structural parameter than the tree-depth. In addition to that, we show \Whness of the problem with respect to the feedback-edge set, the distance to disjoint paths, and \NPhness with respect to the shrub-depth and the clique-width. On a positive note, we propose several novel fixed-parameter algorithms for various parameters that are bounded for dense graphs.
\end{abstract}

\section{Introduction}

A partition of a set $V$ into $p\in\N$ parts is a set ${\pi=\{V_1,\ldots,V_p\}}$ of subsets of~$V$ such that for every $i,j\in [p]\colon V_i \cap V_j = \emptyset$ and $\bigcup_{i=1}^p V_i = V$. In the \ECP problem, we are given an undirected graph~${G=(V,E)}$ together with an integer $p$, and our goal is to partition the vertex set $V$ into $p$ \emph{parts} such that the sizes of each two parts differ by at most $1$ and each class induces a connected sub-graph. Formally, our problem is defined as follows.

\defProblemQuestion{\ECP (\ECPshort)}%
	{A simple undirected and connected $n$-vertex graph~$G = (V,E)$ and a positive integer~${p\in\N}$.}%
	{Is there a partition $\pi=\{V_1, \ldots, V_p\}$ of $V$ such that every part $G[V_i]$ is connected, and $||V_i|-|V_j||\leq 1$ for every pair $i,j \in [p]$?\vspace{0.4em}}

The \ECP problem naturally arises in many fields such as redistricting theory~\cite{Altman1997,LevinF2019,Williams1995}, which is a subfield of computational social choice theory, VLSI circuit design~\cite{BhattTL1984}, parallel computing~\cite{ArbenzLMMS2007}, or image processing~\cite{LucertiniPS1993}, to name a few.

One of the most prominent problems in the graph partitioning direction is the \textsc{Bisection} problem, where our goal is to split the vertex set into two parts $A$ and $B$, each part of size at most $\lceil\frac{n}{2}\rceil$, such that the number of edges between $A$ and $B$ is at most some given $k\in\N$. \textsc{Bisection} is \NPh~\cite{GareyJ1979} even if we restrict the input to unit disc graphs~\cite{DiazM2017} and is heavily studied from the parameterised complexity perspective; see, e.g.,~\cite{vanBevernFSS2015,BuiP1992,CyganLPPS2019,FominGLS2014,Wiegers1990}.
The natural generalisation of the \textsc{Bisection} problem is called \textsc{Balanced Partitioning} where we partition the vertices into $p\in\N$ parts, each of size at most $\lceil\frac{n}{p}\rceil$. \textsc{Balanced Partitioning} is \NPh already on trees~\cite{FeldmannF2015} and a disjoint union of cliques~\cite{AndreevR2006}. The parameterised study of this problem is due to Ganian and Obdržálek~\cite{GanianO2013} and van Bevern et al.~\cite{vanBevernFSS2015}. 
In all the aforementioned problems, we are given only the upper-bound on the size of each part; hence, the parts are not necessarily equitable. Moreover, there is no connectivity requirement for the parts.
For a survey of graph partitioning problems, we refer the reader to the monograph of \mbox{Bulu{\c{c}} et al.}~\cite{BulucMSSS2016}.

On the equitability side, the most notable direction of research is the \textsc{Equitable $k$-Colouring} problem (\textsc{EC} for short). Here, we are given an undirected graph~$G$ and the goal is to decide whether there is a proper colouring of the vertices of $G$ using $k$ colours such that the sizes of each two colour classes differ by at most one. Note that the graph induced by each colour class is necessarily an independent set, and hence is disconnected. As the \textsc{$k$-Colouring} problem can be easily reduced to the \textsc{Equitable $k$-Colouring}, it follows that \textsc{EC} is \NPh. Polynomial-time algorithms are known for many simple graph classes, such as graphs of bounded tree-width~\cite{BodlaenderF2005,ChenL1994}, split graphs~\cite{ChenKL1996}, and many others~\cite{FurmanczykK2005}. The parameterised study was, to the best of our knowledge, initiated by Fellows et al.~\cite{FellowsFLRSST2011} and continued in multiple subsequent works~\cite{EncisoFGKRS2009,FialaGK2011,GomesGS2022}. For a detailed survey of the results on \textsc{EC}, we refer the reader to the monograph by Lih~\cite{Lih2013}.

The \ECP problem then naturally brings the concepts of equitability and connectivity of the vertex set together. It is known that \ECPshort is \NPc~\cite{Altman1997}. Moreover, the problem remains \NPc even if $G$ is a planar graph or for every fixed $p$ at least $2$~\cite{DyerF1984,GareyJ1979}. Enciso et al.~\cite{EncisoFGKRS2009} were the first who studied \ECPshort from the viewpoint of parameterised complexity. They showed that \ECPshort is fixed-parameter tractable with respect to the vertex cover number and the maximum leaf number of $G$.
On the negative side, they showed that it is \Wh to decide the problem for the combined parameter the path-width, the feedback-vertex set, and the number of parts $p$. Moreover, they gave an \XP algorithm for \ECPshort parameterised by tree-width. Later, Gima et al.~\cite{GimaHKKO2022} showed that the problem is fixed-parameter tractable when parameterised by the vertex-integrity of $G$. We give a brief introduction to the parameterised complexity and structural parameters in \Cref{sec:preliminaries}. A more general variant with parametric lower- and upper-bounds on the sizes of parts was studied by Ito et al.~\cite{ItoZN2006}, and Blažej et al.~\cite{BlazejGKPSS2023} very recently introduced the requirement on the maximum diameter of each part.

It is worth pointing out that \ECP is also significant from a theoretical point of view. Specifically, this problem is a very common starting point for many \Whness reductions; see, e.,g.,~\cite{vanBevernFSS2015,BlazejGKPSS2023,DeligkasEGHO2021,MeeksS2020}. Surprisingly, the graph in multiple of the before-mentioned reductions remains the same as in original instance, and therefore our study directly strengthens the results obtained in these works. Since the complexity picture with respect to structural parameters is rather incomplete, many natural questions arise. For example, what is the parameterised complexity of \ECPshort when parameterised by the tree-depth of~$G$? Or, is \ECPshort in \FPT when parameterised by the feedback-edge set? Last but not least, is the problem easier to decide on graphs that are from the graph-theoretical perspective dense, such as cliques?

\subsection{Our Contribution}
In our work, we continue the line of study of the \ECP problem initiated by Enciso et al.~\cite{EncisoFGKRS2009} almost 15 years ago. For an overview of our results, we refer the reader to \Cref{fig:results}; however, we believe that our contribution is much broader. We try to summarise this in the following four points.

\begin{figure}[tb!]
	\centering
	\begin{tikzpicture}
		[every node/.style={
			draw=none,
			fill=gray!10,
			minimum height=1.5em,
			align=center,
			text width=5.5em,
			font = {\scriptsize}
		}]
		
		\tikzstyle{FPT} = [fill=green!15]
		\tikzstyle{Wh} = [fill=orange!30]
		\tikzstyle{NPh} = [fill=red!50]
		\tikzstyle{result} = [draw=black,very thick]
		\tikzstyle{improved} = [dashed]
		
		\node[FPT] (vc) at (-4.8,0.1) {vertex cover \#};
		\node[FPT] (ml) at (0.4,0.1) {max leaf \#};
		\node[FPT,result] (dc) at (-8.2,0.1) {dist. to clique};
		
		\node[FPT,result,improved] (3pvc) at (-2,-1) {3-path vertex cover \#};
		\node[FPT,result] (tc) at (-5.7,-0.9) {twin cover \#};
		\node[FPT,result] (nd) at (-9.3,-1) {neighborhood diversity};
		
		\node[FPT,result,improved] (vi) at (-4.5,-2.2) {vertex\\ integrity \#};
		\node[Wh,result] (4pvc) at (-2,-2.2) {4-path vertex cover \#};
		\node[result,Wh] (fes) at (1.3,-2.2) {feedback-edge set \#};
		\node[fill=none] at (2.2,-2.0) {$\star$};
		\node[fill=none] at (-1.05,-2.0) {$\star$};
		
		\node[Wh,result] (td) at (-3,-3.3) {tree-depth};
		\node[Wh,result] (ddp) at (-0.1,-3.4) {dist. to disj. paths};
		\node[Wh,result] (dcg) at (-6.5,-3.3) {dist. to cluster};
		\node[fill=none] at (-2.1,-3.2) {$\star$};
		\node[fill=none] at (0.8,-3.2) {$\star$};
		
		\node[Wh] (pw) at (-2.7,-4.4) {path-width};
		\node[Wh] (fvs) at (1.3,-4.5) {feedback-vertex set \#};
		\node[result,NPh] (sd) at (-5.7,-4.4) {shrub-depth};
		\node[FPT,result] (mw) at (-8.8,-4.4) {modular-width};
		\node[fill=none] at (-1.8,-4.3) {$\star$};
		\node[fill=none] at (2.2,-4.3) {$\star$};
		
		\node[Wh,result] (tw) at (-0.7,-5.4) {tree-width};
		\node[fill=none] at (0.2,-5.3) {$\star$};
		
		\node[result,NPh] (cw) at (-6,-6.3) {clique-width};
		
		\draw[->] (vc) -- (3pvc.north);
		\draw[->] (vc) -- (tc.north);
		\draw[->] (vc) -- (nd.north);
		\draw[->] (dc) -- (dcg);
		\draw[->] (3pvc) -- (4pvc);
		\draw[->] (3pvc.south east) -- (ddp);
		\draw[->] (3pvc) -- (vi);
		\draw[->] (tc) -- (mw.north);
		\draw[->] (tc) -- (dcg);
		\draw[->] (4pvc) -- (td);
		\draw[->] (nd) -- (mw);
		\draw (mw) -- (cw);
		\draw[->] (vi) -- (td);
		\draw[->] (td) -- (pw);
		\draw[->] (td) -- (sd);
		\draw[->] (pw) -- (tw);
		\draw[->] (ml) -- (fes);
		\draw[->] (ml) -- (ddp);
		\draw[->] (dcg) -- (sd);
		\draw[->] (ddp) -- (pw.north east);
		\draw[->] (ddp) -- (fvs);
		\draw[->] (fes) -- (fvs);
		\draw[->] (fvs) -- (tw);
		\draw[->] (sd) -- (cw);
		\draw[->] (tw) -- (cw);
	\end{tikzpicture}
	\caption{An overview of our results. The parameters for which the problem is in \FPT are coloured green, the parameters for which \ECPshort is \Wh and in \XP have an orange background, and \paraNPh combinations are highlighted in red. Arrows indicate generalisations; e.g., modular width generalises both neighbourhood diversity and twin-cover number. The solid thick border represents completely new results, and the dashed border represents an improvement of previously known algorithm. All our \Whness results hold even when the problem is additionally parameterized by the number of parts~$p$; however, the results marked with~$\star$ becomes fixed-parameter tractable if the size of a larger part $\lceil n/p \rceil$ is an additional parameter.}
	\label{fig:results}
\end{figure}
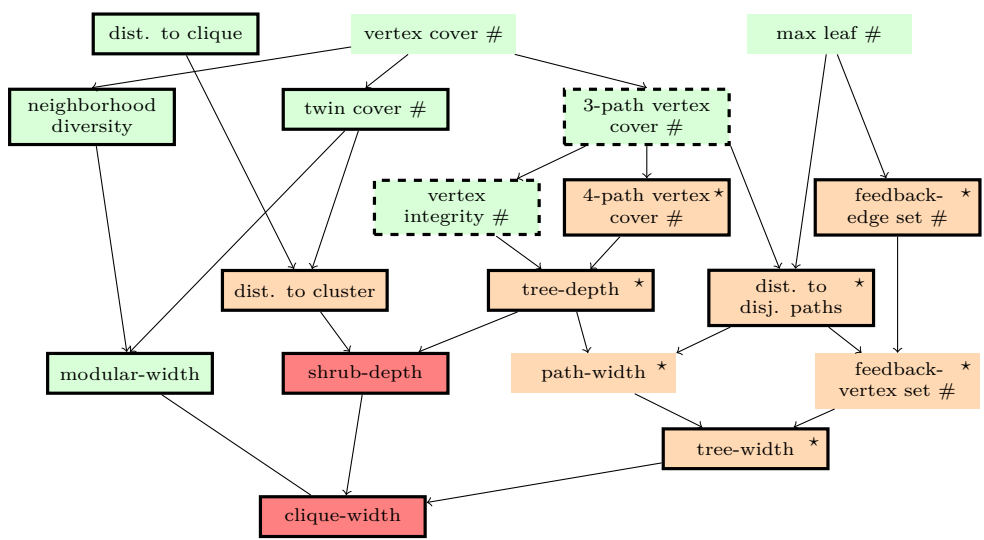

First, directly following pioneering work on structural parameterisation of \ECPshort, we provide a complete dichotomy between tractable and intractable cases for structural parameters that are bounded for sparse graphs. Namely, we provide \Whness proofs for \ECPshort with respect to the $4$-path vertex cover number and the feedback-edge set number, which encloses a gap between structural parameters that were known to be tractable -- the vertex-cover number and the max-leaf number -- and those that were known to be \Wh -- the path-width and the feedback-vertex set. It should also be mentioned that our constructions not only give much stronger intractability results but, at the same time, are much simpler compared to the original construction of Enciso et al.~\cite{EncisoFGKRS2009}.

Second, we also turn our attention to dense graphs, which have, so far, been completely overlooked in the relevant literature. On our way to fixed-parameter tractable algorithms for various structural parameters, we prove polynomial-time solvability of some specific graph classes. Again, we provide a clear boundary between tractable and intractable cases. However, it turns out that for dense graphs, the problem is much easier from a computational perspective.

Third, we clearly show where the limits of the parameterized complexity framework in the study of structural parameterisation of \ECP are. In particular, we show that the problem is \NPh already on graphs of shrub-depth equal to $3$, clique-width equal to $3$, and twin-width equal to $2$. Moreover, in some cases, our complexity results are tight. For example, we give a polynomial-time algorithm for graphs of clique-width~$2$.%

Last but not least, in order to provide all the algorithms, we use multiple different techniques. Naturally, some algorithms are based on standard techniques such as dynamic-programming over decomposition or kernelisation; however, many of them still require deep insights into the structure of the solution and the instances. However, some of them use very careful branching together with formulation of the problem using $N$-fold integer linear programming, which is, informally speaking, an integer linear program with specific shape of the constraints. We are convinced that the technique of $N$-fold integer linear programming in the design and analysis of fixed-parameter tractable algorithms deserves more attention from the parameterised complexity community, as it in many scenarios significantly beats the classical \FPT algorithms based on the famous Lenstra's algorithm; see, e.g., \cite{AltmanovaKK2019,Bredereck0KN2019,KnopKLMO2023,KnopKM2020,KnopSS2022}.

\subsection{Paper Organization}
We start by introducing necessary notation and presenting the parameters under study in \Cref{sec:preliminaries}. Next, we delve into our algorithmic results in \Cref{sec:algorithms}. In \Cref{sec:hardness}, we provide all hardness lower-bounds showing that none of our algorithms can be significantly strengthened. Finally, we conclude in \Cref{sec:conclusions} by discussing possible directions for future research.

\section{Preliminaries}\label{sec:preliminaries}

\subsection{Graph Theory}
For graph-theoretical notation, we follow the monograph by Diestel~\cite{Diestel2017}. A simple undirected \emph{graph} $G$ is a pair $(V,E)$, where $V$ is a non-empty set of \emph{vertices} and $E\subseteq\binom{V}{2}$ is a set of \emph{edges}. We set $n=|V|$ and $m=|E|$. A graph is \emph{connected} if for every pair of vertices $u,v\in V$ there exists a sequence of distinct vertices $v_1,\ldots,v_\ell$ such that $v_1 = u$, $v_\ell = v$, and $\{v_i,v_{i+1}\}\in E$ for every $i\in[\ell-1]$. All graphs assumed in this paper are connected.
Given a vertex $v\in V$, the \emph{neighbourhood} of $v$ is $N(v) = \{u\mid \{u,v\}\in E\}$. The size of the neighbourhood of a vertex $v$ is its \emph{degree} and we denote it as $\deg(v) = |N(v)|$.

\subsection{Parameterised Complexity}
The parameterised complexity framework~\cite{CyganFKLMPPS2015,DowneyF2013,Niedermeier2006} provides tools for a finer-grained complexity analysis of \NPh problems. Specifically, the input of the problem is accompanied by some number $k$, called a \emph{parameter}, that provides additional information about the input data. The best possible outcome for such a parameterised problem is an algorithm with running time $f(k)\cdot n^\Oh{1}$, where $f$ is any computable function. We call this algorithm \emph{fixed-parameter tractable} and the complexity class \FPT contains all problems admitting fixed-parameter tractable algorithms. Less desired but still positive is an algorithm with running time $n^{f(k)}$, where $f$ is again a computable function. A class containing all problems that admit an algorithm with this running time is called \XP.
Naturally, not all parameterised problems are in \FPT. Giving a parameterised reduction from a problem which is known to be \Wh[t] for any $t\geq 1$ to our problem of interest is assumed to be a strong evidence that the problem is not in \FPT. Similarly, one can exclude the existence of an \XP algorithm for some parameterised problem by showing that this problem is \NPh already for a constant value of the parameter $k$. Then, we say that such a problem is \paraNPh.

\subsection{Structural Parameters}
In this sub-section, we provide definitions of all structural parameters we study in this work.

\begin{definition}[$d$-path vertex cover]
	Let $G=(V,E)$ be an undirected graph and $d\in\N$ be an integer. A \emph{$d$-path vertex cover} is a set $C\subseteq V$ such that the graph $G\setminus C$ contains no path with $d$ vertices as a sub-graph. The \emph{$d$-path vertex cover number} $\pvcn[d](G)$ is the size of a minimum $d$-path vertex cover in $G$.
\end{definition}

Note that the $2$-path vertex cover number is, in fact, the standard \emph{vertex cover number} of a graph.

\begin{definition}[Feedback-edge set]
	Let $G=(V,E)$ be a graph. A set $F\subseteq E$ is called \emph{feedback-edge set} of the graph $G$, if $(V,E\setminus F)$ is an acyclic graph. The \emph{feedback-edge set number} $\fes(G)$ is the size of a minimum feedback-edge set in $G$.
\end{definition}

\begin{definition}[Vertex integrity]
	Let $G=(V,E)$ be an undirected graph. \emph{Vertex integrity}, denoted $\vi(G)$, is the minimum $k\in\N$ such that there is a set $X\subseteq V$ of size at most $k$ and every connected component $C$ of $G\setminus X$ contains at most $k$ vertices.
\end{definition}

\begin{definition}[Distance to $\mathcal{G}$]\label{def:distToG}
	Let $G$ be a graph, and $\mathcal{G}$ be a graph family. A set $M\subseteq V(G)$ is a \emph{modulator} to $\mathcal{G}$, if $G\setminus M\in\mathcal{G}$. The \emph{distance to $\mathcal{G}$}, denoted $\operatorname{dist}_{\mathcal{G}}(G)$, is the size of a minimum modulator to~$\mathcal{G}$.
\end{definition}

In our paper, we will focus on the \emph{distance to clique}, denoted by $\operatorname{dc}(G)$, the distance to disjoint union of cliques, which is usually referred to as the \emph{distance to cluster graph} or the \emph{cluster vertex deletion}~\cite{DouchaK2012} and denoted $\operatorname{dcg}(G)$, and the \emph{distance to disjoint paths}, denoted by $\operatorname{ddp}(G)$. %

\begin{definition}[Twin-cover~\cite{Ganian2015}]
	Let $G=(V,E)$ be a graph. A subset $X\subseteq V$ is called a \emph{twin-cover} of $G$ if for every edge $\{u,v\} \in E$ either $u\in X$ or $v\in X$, or $u$ and $v$ are twins.
	We say that $G$ has \emph{twin-cover number} $\operatorname{tc}(G) = k$ if $k$ is the minimum possible size of $X$ for~$G$.
\end{definition}

\begin{definition}[Neighbourhood-diversity~\cite{Lampis2012}]\label{def:nd}
	Let $G=(V,E)$ be a graph. We say that two vertices $u,v\in V$ have the same type iff $N(u)\setminus\{v\} = N(v)\setminus\{u\}$. The \emph{neighbourhood diversity} of $G$ is at most $d$, if there exists a partition of $V$ into at most $d$ sets such that all vertices in each set have the same type.
\end{definition}

Let $T_1,\ldots,T_d$ be a partition of $V$ such that for each $u,v\in T_i$, $i\in[d]$, it holds that $u$ and $v$ are of the same type according to \Cref{def:nd}. Observe that each type is either independent set or a clique. We define \emph{type graph} to be an undirected graph with vertices being the types $T_1,\ldots,T_d$ and two vertices corresponding to some types $T_i$ and $T_j$ are connected by an edge iff there exists an edge $\{u,v\}\in E(G)$ such that $u\in T_i$ and $v\in T_j$. In fact, it is easy to see that if there exists an edge $\{T_i,T_j\}$ in the type graph, then there is an edge $\{u,v\}\in E(G)$ for each pair $u\in T_i$ and $v\in T_j$.´

\begin{definition}[Modular-width~\cite{GajarskyLO2013}]
	Consider graphs that can be obtained from an algebraic expression that uses only the following operations:
	\begin{enumerate}
		\item create an isolated vertex,
		\item the disjoint union of two disjoint graphs $G_1$ and $G_2$ which is a graph $(V(G_1)\cup V(G_2), E(G_1) \cup E(G_2))$,
		\item the complete join of two disjoint graphs $G_1$ and $G_2$ which produces a graph $(V(G_1)\cup V(G_2), E(G_1) \cup E(G_2) \cup \{\{u,v\}\mid u\in V(G_1)\text{ and }v\in V(G_2)\})$.
		\item the substitution with respect to some pattern graph $P$ -- for a graph $P$ with vertices $p_1,\ldots,p_\ell$ and disjoint graphs $G_1,\ldots,G_\ell$, the substitution of the vertices of $P$ by the graphs $G_1,\ldots,G_\ell$ is the graph with vertex set $\bigcup_{i=1}^\ell V(G_i)$ and edge set $\bigcup_{i=1}^\ell E(G_i) \cup \left\{\{u,v\}\mid u\in V(G_i),\ v\in V(G_j),\ \text{and } \{p_i,p_j\}\in E(P)\right\}$. 
	\end{enumerate}
	The width of such an algebraic expression is the maximum number of operands used by any occurrence of the substitution operation. The \emph{modular-width} of a graph $G$, denoted $\mw(G)$, is the least integer $m$ such that $G$ can be obtained from such algebraic expression of width $m$.%
\end{definition}

\begin{definition}[Tree-depth]
	The \emph{tree-depth} $\td(G)$ of a graph $G=(V,E)$ is defined as
	\[
	\td(G) = \begin{cases}
		1 & \text{if $|V| = 1$,}\\
		1 + \min_{v\in V} \td(G - v) & \text{if $G$ is connected and $|V| \geq 2$,}\\
		\max_i \td(C_i) & \text{where $C_i$ is a connected component of $G$, otherwise.}
	\end{cases}
	\]
\end{definition}

\begin{definition}[Shrub-depth~\cite{GanianHNOMR2012}]
	A graph $G=(V,E)$ has a \emph{tree-model} of $m$ colours and depth $d \geq 1$ if there exists a rooted tree $T$ such that
	\begin{enumerate}
		\item $T$ is of depth exactly $d$,
		\item the set of leaves of $T$ is exactly $V$,
		\item each leaf of $T$ is assigned one of $m$ colours,
		\item existence of each edge $e\in E$ depends solely on the colours of its endpoints and the distance between them in $T$.
	\end{enumerate}
	The class of all graphs having a tree-model of $m$ colours and depth $d$ is denoted by $\mathcal{TM}_m(d)$. A class of graphs $\mathcal{G}$ has \emph{shrub-depth} $d$ if there exists $m$ such that $\mathcal{G}\subseteq\mathcal{TM}_m(d)$, while for all natural $m'$ we have $\mathcal{G}\not\in \mathcal{TM}_{m'}(d-1)$.
\end{definition}

\begin{definition}[Tree-width]
	A \emph{tree-decomposition} of a graph $G$ is a pair $\mathcal{T}=(T,\beta)$, where $T$ is a tree and $\beta\colon V(T)\to 2^V$ is a function associating each node of $T$ with a so-called \emph{bag}, such that the following conditions hold
	\begin{enumerate}
		\item $\forall v\in V(G)\colon\exists x\in V(T)\colon v\in \beta(x)$,
		\item $\forall \{u,v\}\in E(G)\colon \exists x\in V(T)\colon \{u,v\}\subseteq \beta(x)$, and
		\item $\forall v\in V(G)$ the nodes $x\in V(T)$ such that $v\in\beta(x)$ induce a connected sub-tree of $T$.
	\end{enumerate}
	The width of a tree-decomposition $\mathcal{T}$ is $\max_{x\in V(T)} |\beta(x) - 1|$. The \emph{tree-width} of a graph $G$, denoted $\tw(G)$, is a minimum width of a decomposition over all tree-decomposition of $G$.
\end{definition}

Let $x$ be a node of a tree-decomposition, by $G^x$ we done the graph $(V^x,E^x)$, where $V^x$ and $E^x$ are the set of all vertices and edges, respectively, introduced in the sub-tree of the tree-decomposition rooted in $x$.

For the purpose of algorithm design and analysis, it is often beneficial to work with a slightly modified tree decomposition with a special structure.

\begin{definition}[Nice tree-decomposition]
	A tree-decomposition $\mathcal{T}=(T,\beta)$ is called \emph{nice tree-decomposition} if and only if
	\begin{itemize}
		\item $T$ is rooted in one vertex $r$ such that $\beta(r) = \emptyset$,
		\item for every leaf $l\in V(T)$ it holds that $\beta(l) = \emptyset$, and
		\item every internal node $x$ is one of the following type:
		\begin{description}
			\item[Introduce vertex node] with exactly one child $y$ such that $\beta(x)\setminus\beta(y) = \{v\}$ for some vertex $v\in V(G)$,
			\item[Forget vertex node] with exactly one child $y$ such that $\beta(y)\setminus\beta(x) = \{v\}$ for some vertex $v\in V(G)$, or
			\item[Join node] with exactly two children $y$ and $z$ such that $\beta(x)=\beta(y)=\beta(z)$.
		\end{description}
	\end{itemize}
\end{definition}

It is known that any tree-decomposition can be turned into a nice one of the same width in $\Oh{\tw^2\cdot n}$ time~\cite{CyganFKLMPPS2015}.

\begin{definition}[Clique-width]
	The \emph{clique-width} of the graph $G$ is the minimum number of labels needed to construct $G$ using solely the following set of operations:
	\begin{enumerate}
		\item Creation of a new vertex $v$ labelled $i$.
		\item Disjoint union of two graphs $G$ and $H$.
		\item Joining by an edge every vertex labelled $i$ to every vertex labelled $j$.
		\item Replacement of label $i$ with label $j$.
	\end{enumerate}
\end{definition}

\subsection{N-fold Integer Programming}
In recent years, integer linear programming (ILP) has become a very useful tool in the design and analysis of fixed-parameter tractable algorithms~\cite{GavenciakKK2022}. One of the best known results in this line of research is probably Lenstra's algorithm, roughly showing that ILP with bounded number of variables is solvable in \FPT time~\cite{Lenstra1983}.

In this work, we use the so-called \emph{$N$-fold integer programming} formulation. Here, the problem is to minimise a linear objective over a set of linear constraints with a very restricted structure. In particular, the constraints are as follows. We use $x^{(i)}$ to denote a set of~$t_i$ variables (a so-called \emph{brick}).
\begin{align}
	D_1 x^{(1)} + D_2 x^{(2)} + \cdots + D_N x^{(N)} &= \textbf{b}_0    \label{eq:NFold:linking}  \\
	A_i x^{(i)}                        &= \textbf{b}_i    & \forall i \in [N]       \\
	\textbf{0} \le x^{(i)}                 &\le \textbf{u}_i  & \forall i \in [N]
\end{align}
Where we have $D_i \in \Z^{r \times t_i}$ and $A_i \in \Z^{s_i \times t_i}$; let us denote $s = \max_{i \in [N]} s_i$, $t = \max_{i \in [N]} t_i$, and let the dimension be~$d$, i.e., $d = \sum_{i \in [N]} t_i \le Nt$.
Constraints~\eqref{eq:NFold:linking} are the so-called \emph{linking constraints} and the rest are the \emph{local constraints}. In the analysis of our algorithms, we use the following result of Eisenbrand et al.~\cite{EisenbrandHKKLO19}.

\begin{proposition}[{\cite[Corollary~91]{EisenbrandHKKLO19}}]\label{prop:n_fold_algo}
	$N$-fold IP can be solved in $a^{r^2s+rs^2} \cdot d \cdot \log(d) \cdot L$ time, where~$L$ is the maximum feasible value of the objective and $a= r\cdot s \cdot \max_{i \in [N]} \left( \max ( \|D_i\|_\infty, \|A_i\|_\infty ) \right)$.
\end{proposition}

\section{Algorithmic Results}\label{sec:algorithms}

In this section, we provide our algorithmic results. The first algorithm is for \ECPshort parameterised by the vertex-integrity and combines careful branching with $N$-fold integer programming. Specifically, Gima et al.~\cite{GimaHKKO2022} showed that \ECPshort is in \FPT with respect to this parameter by giving an algorithm running in $k^{k^{k^\Oh{k}}}\cdot n^\Oh{1}$ time, where $k=\vi(G)$. We show that using an $N$-fold IP formulation, we can give a simpler algorithm with a doubly exponential improvement in the running time.

\begin{theorem}\label{thm:ECP:FPT:vi}
	The \ECP problem is fixed-parameter tractable parameterised by the vertex-integrity \vi and can be solved in $k^\Oh{k^4}\cdot n\log n$ time, where $k = \vi(G)$.
\end{theorem}
\begin{proof}
	First, if it holds that $p > k$, we use the algorithm of Gima et al.~\cite{GimaHKKO2022} which, in this special case, runs in time $k^\Oh{k^2} \cdot n$. Therefore, the bottleneck of their approach is clearly the case when $p \leq k$. In what follows, we introduce our own procedure for this case, which is based on the N-fold integer programming. Note that the algorithm of Gima et al.~\cite{GimaHKKO2022} for the case $p \leq k$ is based on the algorithm of Lenstra~\cite{Lenstra1983}.
	
	First, we guess (by guessing we mean exhaustively trying all possibilities) a partition of the modulator vertices~$X$ in the solution. Let this solution partition be $X_1, \ldots, X_p$. Furthermore, we guess which (missing) connections between the vertices in the modulator will be realised through the components of $G-X$. Let $E(X)$ be the set of these guessed connections.
	
	Now, we check the validity of our guess using the (configuration) $N$-fold ILP. 
	Each component of $G-X$ (call them \emph{pieces}) has at most $k$ vertices; therefore, it can be split in at most $k$ chunks (not necessarily connected) that will be attached to some modulator vertices already assigned to the parts of the solution. Let $\pieces(G,X)$ be the set of all pieces of $G-X$.
	Now, we want to verify if there exists a selection of chunks for every piece so that when we collect these together the solution is indeed connected and contains the right number of vertices.
	Thus, there are altogether at most $k^k$ configurations of chunks in a piece.
	Let $\config(Z)$ be the set of all configurations of a piece~$Z$.
	Let $s^Z_{C,i}$ be the number of vertices in the chunk attached to the $i$-th part from a piece~$Z$ in the configuration~$C$.%
	Let $Z$ be a piece and $C \in \config(Z)$, we set $e^Z_{C}(u,v) = 1$ if the chunk assigned by~$C$ to the part containing both $u,v \in X$ connects $u$ and~$v$.
	
	Now, we have to ensure (local constraint) that each piece is in exactly one configuration
	\begin{equation}\label{eq:ECP:FPT:vi:NFold:local}
		\sum_{C \in \config(Z)} x^Z_C = 1 \qquad\qquad \forall Z \in \pieces(G,X) \,.
	\end{equation} %
	Observe that these constraints have no variables in common for two distinct elements of~$\pieces(G,X)$.
	The rest of the necessary computation uses global constraints.
	We ensure that the total contribution of chunks assigned to the parts is the correct number ($x_i$ is a binary slack such that $\sum_i x_i = n \bmod p$): %
	\begin{equation}\label{eq:ECP:FPT:vi:NFold:chunk_contrib}
		x_i + \sum_{Z \in \pieces(G,X)}\sum_{C \in \config(Z)} s^Z_{C,i} \cdot x^Z_C
		=
		\lceil n/p \rceil - |X_i|
		\qquad\qquad \forall i \in [p]
	\end{equation}%
	Next, we have to verify the connectivity of parts in $X$
	\begin{equation}\label{eq:ECP:FPT:vi:NFold:connectivity}
		\sum_{Z \in \pieces(G,X)}\sum_{C \in \config(Z)} e^Z_{C}(u,v) \cdot x^Z_C
		\ge 1
		\qquad\qquad \forall \{u,v\} \in E(X)
	\end{equation}
	
	It is not hard to verify, that the parameters of $N$-fold IP are as follows:
	\begin{itemize}
		\item the number $s$ of local constraints in a brick is exactly $1$ as there is a single local constraint~\eqref{eq:ECP:FPT:vi:NFold:local} for each piece,
		\item the number $r$ of global constraints is in $\Oh{k^2}$: there are $p \le k$ constraints \eqref{eq:ECP:FPT:vi:NFold:chunk_contrib} and $\binom{k}{2}$ constraints \eqref{eq:ECP:FPT:vi:NFold:connectivity},
		\item the number $t$ of variables in a brick is $|\config(Z)|$ which is~$k^{\Oh{k}}$, and
		\item $a \in \Oh{k^3}$, since all coefficients in the constraints are bounded by~$k$ in absolute value.
	\end{itemize}
	Thus, using \Cref{prop:n_fold_algo}, the \ECP problem can be solved in $k^\Oh{k^4} \cdot n \log n$ time.
\end{proof}

Using techniques from the proof of \Cref{thm:ECP:FPT:vi}, we may give a specialised algorithm for \ECP parameterised by the $3$-path vertex cover. The core idea is essentially the same, but the components that remain after removing the modulator are much simpler: they are either isolated vertices or isolated edges. This fact allows us to additionally speed the algorithm~up.

\begin{theorem}
	\label{thm:ECP:FPT:3pvc}
	The \ECP{} problem is fixed-parameter tractable parameterised by the $3$-path vertex cover number and can be solved in $k^\Oh{k^2}\cdot n\log n$ time, where $k = \pvcn[3](G)$.
\end{theorem}
\begin{proof}
	We prove the theorem by investigating separately the case when $p > k$ and $p \leq k$. In particular, we show that in the former case, the problem can be solved in \FPT time relatively trivially, while for the later case, we use an $N$-fold integer programming machinery similar to \Cref{thm:ECP:FPT:vi}. For the rest of this proof, let $\ell = (n \bmod p)$ stands for the number of large parts and $s = p - \ell$ be the number of small parts.
	
	First, assume that $p > k$. 
	If $\lfloor n/p \rfloor \geq 3$, then we are dealing with clear no instance as there exists at least one part not intersecting $M$; however, the size of each component of $G\setminus M$ is at most $2$. Therefore, it holds that $\lfloor n/p \rfloor \leq 2$.
	If $\lfloor n/p \rfloor = 0$, the problem is trivial, and if $\lfloor n/p \rfloor = 1$, we can solve the problem in polynomial time by using an algorithm for maximum matching. 
	The only remaining case occurs when $\lfloor n/p \rfloor = 2$. First, if $\ell > k$, then the instance is clearly \emph{no}-instance, as each large part has to contain at least one vertex of the modulator. Therefore, $\ell \leq k$. Additionally, let $I_0 \subseteq V(G)\setminus M$ be the set of vertices such that they are isolated in $G\setminus M$. Each such vertex has to be in part with at least one modulator vertex. Therefore, if $|I_0| > 2k$, we directly return \emph{no}. Hence, we can assume that $|I_0| \leq 2k$. As our first step, we guess a set $E_G\subseteq (E(G) \cap (\binom{M\cup I_0}2\setminus\binom{M}2))$ of edges between modulator and vertices of $I_0$ that secures connectivity for vertices in $I_0$. There are at most $k\cdot 2k$ possible edges and therefore $2^\Oh{k^2}$ possible options we need to check. Now, we guess the set of edges $E_B\subseteq E(G)\cap\binom{M}2$ inside the modulator that are important for connectivity. Again, there are $k$ modulator vertices and each is adjacent to at most $k-1$ edges. This gives us $2^\Oh{k^2}$ possibilities we try. For each such guess, let $V_1,\ldots,V_{k'}$, where $k' \leq k$, be the partition of vertices in $M\cup I_0$ defined such that two vertices $u,v\in M\cup I_0$ are in the same part if and only if $\{u,v\} \in E_G \cup E_B$. For every $V_i$, $i \in [k']$, we check whether $V_i$ is connected, of size at most $3$, and does not consist of a single vertex from $I_0$. If any of the conditions is break, we reject the guess and continue with another option. Otherwise, we guess which parts are of size $3$ and which of them are of size $2$. There are $2^\Oh{k}$ options and, without loss of generality, let $V_1,\ldots,V_\ell$ be the parts for which we guessed that they are large. First, we check that no part is larger than what was guessed. If there is such a part, we reject the guess. Otherwise, for every part $V_i$, $i\in[k']$, that is missing a vertex, we guess at most $2$ (depending on the guessed size of $V_i$) component types in $G\setminus (M\cup I_0)$ whose vertices are in the solution in $V_i$. As there are at most $k$ parts and at most $2^\Oh{k}$ component types, this gives us $2^\Oh{k^2}$ options. In every option, each cluster type $t$ has to occur even times as otherwise we end up with an odd number of vertices used in parts intersecting the modulator and therefore, there remains one isolated vertex. Hence, we assume that all component types are guessed even times. Now, for every component $C \in G\setminus M$ whose vertices are not used in parts intersecting~$M$, we create a new part $V_j$ containing all (two) vertices of this component. Finally, we check that the each part of the guessed partition is connected and has correct size. This can be done in polynomial time. Overall, there are $2^\Oh{k^2}\cdot 2^\Oh{k^2}\cdot 2^\Oh{k} \cdot 2^\Oh{k^2} = 2^\Oh{k^2}$ options and the running time of our algorithm is $2^\Oh{k^2}\cdot n^\Oh{1}$.
	
	Next, let $p \leq k$. 
	We again guess a partition $X_1,\ldots,X_p$ of modulator vertices.
	This time, each piece has at most $2$ vertices and therefore can be split in at most $2$ chunks.
	In order to simplify the ILP we further guess which types of pieces are added to each $X_i$ so that after this all the partitions (chunks) are connected.
	Note that we use at most $k$ pieces for this as this process can be viewed as adding edges to the modular so that the chunks are connected; we need at most $k-1$ new edges for this.
	We add the guessed pieces to $X_i$ and to the modulator; thus, the modulator can grow to in size at most $3k$.
	In total, this takes at most $2^{\Oh{k^2}}$ time, as we first guess $k-1$ piece types (from $2^k$ with repetition) and then assign these $k-1$ pieces to $p$ chunks.
	Then, we introduce nearly the same N-fold ILP as in \Cref{thm:ECP:FPT:vi}, whose correctness follows by the same arguments.
	This time, we omit the constraints \eqref{eq:ECP:FPT:vi:NFold:connectivity}, since our chunks are already connected.
	Now, the parameters of the resulting N-fold ILP are:
	\begin{itemize}
		\item the number of local constraints in each brick is $s = 1$ as there is a single local constraint~\eqref{eq:ECP:FPT:vi:NFold:local} for each type,
		\item the number of global constraints $r \in \Oh{k}$: there are $p \le k$ constraints \eqref{eq:ECP:FPT:vi:NFold:chunk_contrib},
		\item the number of variables in a brick $t$ is at most~2, and
		\item $a \in \Oh{k}$, since all coefficients in the constraints are bounded by~$2$ in absolute value (a contribution of a piece to a chunk is either 0, 1, or 2).
	\end{itemize}
	Therefore, the above algorithm runs in $k^{\Oh{k^2}} \cdot n \log n$ time by \Cref{prop:n_fold_algo}.
\end{proof}

The next result is an \XP algorithm with respect to the tree-width of $G$. It should be noted that a similar result was reported already by Enciso et al.~\cite{EncisoFGKRS2009}; however, their proof was never published\footnote{In particular, in their conference version, Enciso et al.~\cite{EncisoFGKRS2009} promised to include the proof in an extended version, which, however, has never been published. There is also a version containing the appendix of the conference paper available from \url{https://www.researchgate.net/publication/220992885_What_Makes_Equitable_Connected_Partition_Easy}; however, even in this version, the proof is not provided.} and the only clue for the algorithm the authors give in~\cite{EncisoFGKRS2009} is that this algorithm ``can be proved using standard techniques for problems on graphs of bounded treewidth''. Therefore, to fill this gap in the literature, we give our own algorithm.

\begin{theorem}
	\label{thm:ECP:XP:tw}
	The \ECP problem is in \XP when parameterized by the tree-width \tw of $G$.
\end{theorem}
\begin{proof}
	We proceed by a leaf-to-root dynamic programming along a nice tree decomposition. We assume that a decomposition of the optimal width is given as part of the input, as otherwise, we can find one in $2^\Oh{\tw^2}\cdot n^\Oh{1}$ time using the algorithm of Korhonen and Lokshtanov~\cite{KorhonenL2023}. 
	
	The crucial observation we need for the algorithm is that at every moment of the computation, there are at most $\Oh{\tw}$ \emph{opened parts}. This holds because each part needs to be connected and each bag of the tree-decomposition forms a separator in $G$; therefore no edge can ``circumvent'' currently processed bag.
	
	The algorithm then proceeds as follows. In each node $x$ of the tree-decomposition, we try all possible partitions of vertices into opened parts and the number of past vertices for each opened part. Once a new vertex $v$ is introduced, we have three possibilities: create a new part consisting of only~$v$, put~$v$ into some existing opened part, or merge (via $v$) multiple already existing opened parts into a new one. When a vertex~$v$ is forgotten, we need to check whether $v$ is the last vertex of its part and, if yes, whether the part vertex $v$ is member of is of the correct size. In join nodes, we just merge two records from child nodes with the same partition of bag vertices.
	
	More precisely, in each node $x$ of the tree-decomposition, we store a single dynamic programming table $\operatorname{DP}_x[\mathcal{P},\sigma]$, where
	\begin{itemize}
		\item $\mathcal{P}\colon\beta(x)\to[\tw + 1]$ is a partition of bag vertices into opened parts,
		\item $\sigma\colon[\tw + 1]\to[\lceil n/p \rceil]$ is a vector containing for each opened part the number of past vertices from this part.
	\end{itemize}
	The value of every cell is $\mathtt{true}$ iff there exists a \emph{partial solution} $\pi^x_{\mathcal{P},\sigma} = (V_1,\ldots,V_k)$ such that
	\begin{enumerate}
		\item $\bigcup_i^k = V^x$,
		\item $V_i\cap V_j = \emptyset$ for each pair of distinct $i,j\in[k]$,
		\item each $V_i$ is connected,
		\item for each part $V_i$ with $V_i\cap\beta(x) = \emptyset$ (called closed part), the size of $V_i$ is either $\lfloor n/p \rfloor$ or $\lceil n/p \rceil$, and
		\item for each part $V_i$ with $V_i\cap\beta(x) \not= \emptyset$ (opened part) there exists $j\in[|\beta(x)|]$ such that $V_i\cap\beta(x) = \mathcal{P}^{-1}(j)$ and $|V_i| = |\mathcal{P}^{-1}(j)| + \sigma(j)$.
	\end{enumerate}
	Otherwise, the stored value is $\mathtt{false}$. Once the dynamic table is correctly filled, we ask whether the dynamic table for the root of the tree-decomposition stores \texttt{true} in its single cell.
	
	\begin{claim}\label{cl:TW:DP:tableSize}
		The size of the dynamic programming table $\DP_x$ for a node $x$ is $n^\Oh{\tw}$.
	\end{claim}
	\begin{claimproof}
		For a node $x$, there are $|\beta(x)|^\Oh{\tw} = \tw^\Oh{tw}$ possible partitions of bag vertices. For each partition, there are $\lceil n/p \rceil^\Oh{\tw} = n^\Oh{\tw}$ different vectors $\sigma$. This gives us $\tw^\Oh{\tw}\cdot n^\Oh{\tw} = n^\Oh{\tw}$ cells in total.
	\end{claimproof}
	
	Now, we describe how the computation works separately for each node type of the nice tree-decomposition.
	
	\subparagraph{Leaf Node.} Leaf nodes are by definition empty. Therefore, for every leaf node $x$, the only possible partition $\mathcal{P}$ is empty and also $\sigma$ is always a vector of zero length. That is, we set
	\[
	\operatorname{DP}_x[\mathcal{P},\sigma] = \begin{cases}
		\mathtt{true} & \text{if }\forall i \in [\tw + 1]\colon \sigma(i) = 0 \text{ and}\\
		\mathtt{false} & \text{otherwise.}
	\end{cases}
	\]
	since it clearly holds that each part is connected and each closed part is of the correct size.
	
	\begin{claim}\label{lem:leafRunTime}
		The computation in every leaf node can be done $n^\Oh{\tw}$ time.
	\end{claim}
	\begin{claimproof}
		For leaf nodes, it holds by definition that $\beta(x) = \emptyset$. Hence, $\mathcal{P}$ is always an empty set. Consequently, we just try all possible values of $\sigma$ and compute each cell in $\Oh{\tw}$ time. That is, the computation in leaf nodes takes $n^\Oh{\tw} \cdot \tw$ time.
	\end{claimproof}
	
	\subparagraph{Introduce Node.} Let $x$ be an introduce node introducing a vertex $v$ and let $y$ be its child. For a partition $\mathcal{P}$ we denote by $P_v$ the part in which $v$ is assigned in $\mathcal{P}$ and by $i$ the index of the record in $\sigma$ associated with $P_v$.
	\[
	\operatorname{DP}_x[\mathcal{P},\sigma] = \begin{cases}
		\DP_y[ \left.\mathcal{P}\right\vert_{\beta(x)\setminus\{v\}}, \sigma ] 
		& \text{if } \mathcal{P}^{-1}(\mathcal{P}(v)) = \{v\}\\&~~~~~~~~~\text{ and } \sigma(\mathcal{P}(v)) = 0,\\
		\mathtt{false}
		& \text{if } \mathcal{P}^{-1}(\mathcal{P}(v)) = \{v\}\\&~~~~~~~~~\text{ and } \sigma(\mathcal{P}(v)) > 0,\\
		\mathtt{false}
		& \text{if } \mathcal{P}^{-1}(\mathcal{P}(v)) \cap N(v) = \emptyset,\\
		\bigvee\limits_{\substack{
				\mathcal{P}'\colon\beta(x)\to[\tw+1],\ \sigma'\colon[\tw+1]\to[\lceil n/p \rceil]\\
				\forall u \not\in \mathcal{P}^{-1}(\mathcal{P}(v))\colon \mathcal{P}'(u) = \mathcal{P}(u)\\
				\forall w \in \mathcal{P}^{-1}(\mathcal{P}(v))\colon \mathcal{P}'(w) \in \{i\mid\mathcal{P}^{-1}(i) = \emptyset\}\cup\{\mathcal{P}(v)\}\\
				\forall i \in \{j\mid \mathcal{P}(j) = \emptyset\} \cup \{\mathcal{P}(v)\}\exists u\in\mathcal{P}'^{-1}(i)\colon u\in N(v)\\
				\forall i\in\{j\mid \mathcal{P}^{-1}(j) \not=\emptyset\}\setminus\{\mathcal{P}(v)\}\colon \sigma'(i) = \sigma(i)\\
				\forall i \in \{j\mid\mathcal{P}^{-1}(j) = \mathcal{P}'^{-1}(j) = \emptyset\}\colon \sigma'(i) = 0\\
				\sum_{i\in\{j\mid\mathcal{P}^{-1}(j) = \emptyset\}\cup\{\mathcal{P}(v)\}} \sigma'(i) = \sigma(\mathcal{P}(v))\\
				
		}} \DP_y\left[ \mathcal{P}', \sigma' \right] & \text{otherwise.}
	\end{cases}
	\]
	
	\begin{claim}\label{lem:introduceRunTime}
		The computation in every introduce node can be done in $n^\Oh{\tw}$ time.
	\end{claim}
	\begin{claimproof}
		The most time-consuming operation of the computation is clearly the last case of the recurrence. Here, we examine $\tw^\Oh{\tw}$ functions $\mathcal{P}'$ and, for each $\mathcal{P}'$, $n^\Oh{\tw}$ different functions $\sigma'$. Overall, the computation in a single cell takes $n^\Oh{\tw}$ time. By \Cref{cl:TW:DP:tableSize}, there are $n^\Oh{\tw}$ different cells and therefore, the table in introduce nodes can be filled in $n^\Oh{\tw}\cdot n^\Oh{\tw} = n^\Oh{\tw}$ time.
	\end{claimproof}
	
	\subparagraph{Forget Node.} When forgetting a vertex, we try to put it in all the opened parts. Separately, we take care of situations, where $v$ is the last vertex of part.
	\begin{align*}
		\DP_x[ \mathcal{P}, \sigma ] = 
		&\bigvee_{\substack{i\in [\tw+1]\\\mathcal{P}^{-1}(i) \not= \emptyset\\\sigma(i)\geq 1}} \DP_y\left[ \mathcal{P}\Join (v \mapsto i), \sigma \Join \left(i \mapsto (\sigma(i) - 1)\right) \right]\,\lor\\
		&\bigvee_{\substack{i\in[\tw+1]\\\mathcal{P}^{-1}(i) = \emptyset\\\sigma(i) = 0}} \bigvee_{s\in\left\{\lfloor\frac{n}{p}\rfloor,\lceil\frac{n}{p}\rceil\right\}} \DP_y\left[ \mathcal{P} \Join (v \mapsto i), \sigma \Join \left(i \mapsto (s - 1)\right) \right],
	\end{align*}
	where, for a function $f$, we use $f \Join (i\mapsto j)$ to denote function $f'$ which returns the same images as $f$ for all values except for the value $i$, for which the function returns the image $j$. More specifically, if the value $i$ is not in the domain of $f$, it is added and $j$ is $i$'s image in $f'$. If $i$ is in the domain of $f$, the original image $f(i)$ is replaced with $j$.
	
	\begin{claim}\label{lem:forgetRunTime}
		The computation in every forget node can be done $n^\Oh{\tw}$ time.
	\end{claim}
	\begin{claimproof}
		By \Cref{cl:TW:DP:tableSize}, there are $n^\Oh{\tw}$ different cells. The definition of the computation of the single cell, the value of each cell can be determined in $\Oh{\tw}$ time; we just need to try all non-empty parts of $\mathcal{P}$. Overall, the computation in every join node takes $n^\Oh{tw}\cdot \Oh{\tw} = n^\Oh{\tw}$ time.
	\end{claimproof}
	
	\subparagraph*{Join Node.} In join nodes, we may need to merge parts that are disconnected in both $G^y$ and $G^z$; however, become connected using the past vertices of the other subtree; see \Cref{fig:ECP:XP:tw:joinNode} for an illustration. Before we formally define the recurrence for join nodes, let us introduce an auxiliary notation. Let $\mathcal{P}$ be a partition of bag vertices and $\sigma$ be a vector containing the number of past vertices for each part. By $\mathbb{S}(\mathcal{P},\sigma)$, we denote the set of all \emph{splits}: quadruplets $(\mathcal{P}_y,\sigma_y,\mathcal{P}_z,\sigma_z)$ where for each non-empty $\mathcal{P}^{-1}(i)$, there exists $i_1,\ldots,i_\ell$ and $i'_1,\ldots,i'_{\ell'}$ such that $\mathcal{P}^{-1}(i) = \bigcup_{j=1}^{\ell} \mathcal{P}_y^{-1}(i_j) = \bigcup_{j=1}^{\ell'} \mathcal{P}_z^{-1}(i_j)$, $\sigma(i) = \sigma_y(i_1) + \cdots + \sigma_y(i_\ell) + \sigma_z(i_1) + \cdots + \sigma_z(i_{\ell'})$, $\forall j,j'\in[\ell]$ there exists a sequence $j_1,j_2,\ldots,j_k$, where $j_1 = j$, $j_k = j'$, $\forall \alpha\in\{1,3,\ldots,k\}\colon \mathcal{P}^{-1}_y(j_\alpha) \cap \mathcal{P}^{-1}_z(j_{\alpha+1}) \not=\emptyset$, $\forall \beta\in\{2,4,\ldots,k-1\}\colon \mathcal{P}^{-1}_z(j_\beta)\cap\mathcal{P}^{-1}_y(j_{\beta+1}) \not= \emptyset$, and similarly for all pairs of $j,j'\in[\ell']$ and $j\in[\ell]$ and $j'\in[\ell']$. Intuitively, splits describe all possible situations how each part of $\mathcal{P}$ can be created by merging smaller parts in sub-trees while keeping the part connected. The computation is formally defined as follows.
	
	\[
	\DP_x[ \mathcal{P}, \sigma ] = \bigvee_{(\mathcal{P}_y,\sigma_y,\mathcal{P}_z,\sigma_z)\in \mathbb{S}(\mathcal{P},\sigma)} \DP_y[ \mathcal{P}_y, \sigma_y ] \land \DP_z[ \mathcal{P}_z, \sigma_z ].
	\]
	
	\begin{figure}[bt!]
		\centering
		\begin{tikzpicture}
			\draw (0,0) ellipse (2cm and 1cm);
			\node[draw,circle] (v4) at (1.5,0) {$v_4$};
			\node[draw,circle] (v3) at (0.5,0) {$v_3$};
			\node[draw,circle] (v2) at (-0.5,0) {$v_2$};
			\node[draw,circle] (v1) at (-1.5,0) {$v_1$};
			\node at (2.5,0.5) {$\beta(x)$};
			
			\draw (2,0) to[out=-50,in=-80,distance=5cm] (-2,0);
			\draw (-2,0) to[out=-130,in=-100,distance=5cm] (2,0);
			\node at (-3,-3) {$G^y$};
			\node at (3,-3) {$G^z$};
			
			\draw[red,snake it] (-2.5,-2) -- (v1);
			\draw[blue,snake it] (v2) to[out=-140,in=-90,distance=4cm] (v3);
			\draw[orange, snake it] (v2) to[out=-30,in=-70,distance=4.75cm] (v1);
			\draw[green, snake it] (v3) to[out=-70,in=-60,distance=3cm] (v4);
		\end{tikzpicture}
		\caption{An illustration of merging different parts into a new one in join nodes. Suppose that, in $\pi_y$, the vertex $v_1$ is in a red part, vertices $v_2$ and $v_3$ are together in a blue part, and $v_4$ is in singleton part without any past vertices. Moreover, in $\pi_z$, vertices $v_1$ and $v_2$ are in an orange part and $v_3$ and $v_4$ are in a green part. Observe that in $G^y$ and $G^z$, respectively, they cannot form one part as their are completely disconnected in the respective subtrees. However, in $G^x$, we can merge them as the connectivity is now secured via past vertices in the other subtree.}
		\label{fig:ECP:XP:tw:joinNode}
	\end{figure}
	
	\begin{claim}\label{lem:joinRunTime}
		The computation in every join node can be done $n^\Oh{\tw}$ time.
	\end{claim}
	\begin{claimproof}
		By \Cref{cl:TW:DP:tableSize}, there are $n^\Oh{\tw}$ different cells. In the worst case, there are $(\tw^\Oh{\tw}\cdot n^\Oh{\tw})^2 = n^\Oh{\tw}$ differents splits of $\mathcal{P}$ and $\sigma$ and children table look-up for one split is constant operation. Overall, the running time in join nodes is $n^\Oh{\tw}\cdot n^\Oh{\tw} = n^\Oh{\tw}$.
	\end{claimproof}
	
	\medskip
	
	Now, we show that our algorithm is correct. We use two auxiliary lemmas. First, we show that whenever our dynamic programming table stores \texttt{true}, there is a corresponding partial solution.
	
	\begin{lemma}\label{thm:ECP:XP:tw:trueVal}
		For every node $x$, if $\DP_x[ \mathcal{P}, \sigma ] = \mathtt{true}$, then there exists a partial solution~$\pi^x_{\mathcal{P},\sigma}$ compatible with $(\mathcal{P},\sigma)$. 
	\end{lemma}
	\begin{proof}
		We prove the lemma by a bottom-up induction on the tree decomposition of $G$. 
		First, let $x$ be a leaf node. Since $\DP_x[\mathcal{P},\sigma]$ is $\mathtt{true}$, it holds that $\mathcal{P}^{-1}(i) = \emptyset$ and $\sigma(i) = 0$ for every $i\in[\tw+1]$. Clearly, an empty partition $\pi^x_{\mathcal{P},\sigma}$ is a partial solution compatible with $(\mathcal{P},\sigma)$, as $V^x = \emptyset$. Therefore, for leaves the lemma holds. Let us now suppose that the lemma holds for all children of node $x$.
		
		Let $x$ be an introduce node with a single child $y$ introducing vertex $v$. Since $\DP_x[\mathcal{P},\sigma]$ is $\mathtt{true}$, there exist $(\mathcal{P}',\sigma')$ such that $\DP_y[\mathcal{P}',\sigma'] = \mathtt{true}$ which caused that $\DP_x[\mathcal{P},\sigma]$ is $\mathtt{true}$. By the induction hypothesis, there exists a partial solution $\pi^y_{\mathcal{P}',\sigma'}$ compatible with $(\mathcal{P}',\sigma')$. First, suppose that $\mathcal{P} = \mathcal{P}' \Join (v \mapsto i)$ for some $i\in[\tw+1]$ such that $\mathcal{P}'^{-1}(i) = \emptyset$. Then it holds that $\sigma' = \sigma$. We create $\pi^x_{\mathcal{P},\sigma}$ by adding new part containing only the vertex $v$ into $\pi^y_{\mathcal{P}',\sigma'}$. All the copied part all connected and of correct sizes and the new part has no past vertices, which is consistent with $\sigma$. Moreover, part with one vertex is always connected, so $\pi^x_{\mathcal{P},\sigma}$ is correct. In the remaining case, there exist $j_1,\ldots,j_k$ such that $\mathcal{P}^{-1}(\mathcal{P}(v)) = \mathcal{P}'^{-1}(j_1) \cup \cdots \cup \mathcal{P}'^{-1}(j_k)$ and $\mathcal{P}^{-1}(j_1) = \cdots = \mathcal{P}^{-1}(j_k) = \emptyset$, with a possible expectation of $j_\ell = i$. The partition of all remaining bag vertices is the same in both $\mathcal{P}$ and $\mathcal{P}'$. Without loss of generality, let $\mathcal{P}'^{-1}(j_\ell)$ are in part $V_{j_\ell}$ in $\pi^y_{\mathcal{P}',\sigma'}$. For every $j\not\in\{j_1,\ldots,j_k\}$, we add to $\pi^x_{\mathcal{P},\sigma}$ all parts $V_j$ from $\pi^y_{\mathcal{P}',\sigma'}$. Finally, we add one new part $V = \bigcup_{i=1}^k V_{j_i}$. It holds that $|V| = \sigma'(j_1) + \cdots + \sigma'(j_k) = \sigma(i)$, that is, the size of $V$ is correct. For the connetivity, let $u$ and $w$ be two past vertices. If $u,w\in V_{j}$, $j\in[k]$, then they are connected as $V_j$ is connected. Therefore, let $u\in V_j$ and $w\in V_{j'}$ for some distinct $j,j'\in\{j_1,\ldots,j_k\}$. Let $u'\in\beta(x)\cap V_j \cap N(v)$ and $w'\in\beta(x) \cap V_{j'} \cap N(v)$ (such $u'$ and $w'$ exists by the definition of the computation). Clearly, there exists $u,u'$-path and $w',w$-path, so by connecting these two paths via $v$, we obtain that $V$ is indeed connected.
		
		Next, let $x$ be a forget node with a single child $y$ and let $v$ be the forgotten vertex. Suppose that $\DP_x[\mathcal{P},\sigma]$ is $\mathtt{true}$. Then, by the definition of the computation, there exists $(\mathcal{P'},\sigma')$ such that $\DP_y[\mathcal{P}',\sigma'] = \mathtt{true}$ and which caused that $\DP_x[\mathcal{P},\sigma]$ is set to $\mathtt{true}$. By the induction hypothesis, there exists a partition $\pi^y_{\mathcal{P'},\sigma'}$ which is a partial solution compatible with $(\mathcal{P}',\sigma')$. We claim that by setting $\pi^x_{\mathcal{P},\sigma} = \pi^y_{\mathcal{P}',\sigma'}$, we obtain a partial solution compatible with $(\mathcal{P},\sigma)$. Since, $V^x = V^y$, $\pi^x_{\mathcal{P},\sigma}$ clearly partitions the whole $V^x$. Let $V_i$ be a part which $v$ is part of in $\pi^y_{\mathcal{P},\sigma'}$ and, without loss of generality, let $\mathcal{P}'(v) = i$. Clearly, all closed parts remain valid and so remain the opened parts except for $V_i$. First, suppose that $\beta(x)\cap V_i = \emptyset$, that is, $V_i$ is now a closed part. It follows that $\beta(y)\cap V_i = \{v\}$ and also $|V_i| \in \{\lfloor n/p \rfloor,\lceil n/p \rceil\}$ by the definition of the computation. Hence, $\pi^x_{\mathcal{P},\sigma}$ is a partial solution compatible with $(\mathcal{P},\sigma)$. It remains to show the correctness for the case when $\beta(x)\cap V_i \not= \emptyset$. In this case, by the definition of the computation, it holds that $\mathcal{P'} = \mathcal{P} \Join (v \mapsto i)$ and $\sigma' = \sigma \Join (i \mapsto (\sigma(i) - 1))$. Then $V_i$ is still opened part. Connectedness of $V_i$ is preserved and the number of past vertices is $\sigma'(i) + 1$. Therefore, even in this case a partial solution exists, finishing the correctness for forget nodes.
		
		The last remaining case is when $x$ is a join node with exactly two children $y$ and $z$. Since $\DP_x[\mathcal{P},\sigma]$ is $\mathtt{true}$, there exists a split $(\mathcal{P}_y,\sigma_y,\mathcal{P}_z,\sigma_z)$ such that both $\DP_y[\mathcal{P}_y,\sigma_y]$ and $\DP_z[\mathcal{P}_z,\sigma_z]$ are $\mathtt{true}$. By the induction hypothesis, there exist partial solutions $\pi^y_{\mathcal{P}_y,\sigma_y}$ and $\pi^z_{\mathcal{P}_z,\sigma_z}$ compatible with $(\mathcal{P}_y,\sigma_y)$ and $(\mathcal{P}_z,\sigma_z)$, respectively. Recall that since $(\mathcal{P}_y,\sigma_y,\mathcal{P}_z,\sigma_z)$ is a split, it holds that $\forall i\in[\tw+1]$ there exists $i_1,\ldots,i_\ell$ and $i'_1,\ldots,i'_{\ell'}$ such that $\mathcal{P}^{-1}(i) = \bigcup_{j=1}^\ell \mathcal{P}_y^{-1}(i_j) = \bigcup_{j=1}^{\ell'} \mathcal{P}_z^{-1}(i_j)$, there is a non-empty intersection between, and $\sigma(i) = \sigma_y(i_1) + \cdots + \sigma_y(i_\ell) + \sigma_z(i_1) + \cdots + \sigma_z(i_{\ell'})$. We create a partition $\pi^x_{\mathcal{P},\sigma}$ as follows. We copy all closed parts of $\pi^y_{\mathcal{P}_y,\sigma_y}$ and $\pi^y_{\mathcal{P}_y,\sigma_y}$. As all closed parts were of correct size and connected, they remains so also in the partition $\pi^x_{\mathcal{P},\sigma}$. Now, we define the opened parts. Let $i\in [\tw+1]$ be an integer such that $\mathcal{P}^{-1}(i) \not= \emptyset$ and let $i_1,\ldots,i_\ell$ and $i'_1,\ldots,i'_{\ell'}$ are the indices of the corresponding parts in $\mathcal{P}_y$ and $\mathcal{P}_z$. Without loss of generality, we assume that for every $j\in[\ell]$, the vertices of $\mathcal{P}_y^{-1}(i_j)$ are members of $V'_{i_j}$ in $\pi^y_{\mathcal{P}_y,\sigma_y}$, and for every $j\in[\ell']$, the vertices of $\mathcal{P}_z^{-1}(i_j)$ are members of $V''_{i_j}$ in $\pi^z_{\mathcal{P}_z,\sigma_z}$. We add to $\pi^x_{\mathcal{P},\sigma}$ a new part $V$ created as union of all parts corresponding to parts that are merged in $\mathcal{P}(i)$. Formally, $V = \bigcup_{j=1}^{\ell} V'_{i_j} \cup \bigcup_{j=1}^{\ell'} V''_{i_j}$. It holds that $|V| = \sigma_y(i_1) + \cdots + \sigma_y(i_\ell) + \sigma_z(i'_1) + \cdots + \sigma_z(i'_{\ell'}) = \sigma(i)$. Hence, the size of this part is clearly correct. What remains to show is that the part $V$ is connected. Let $u$ and $w$ be two past vertices. First, suppose that $u,w\in V^y$. If additionally both $u$ and $w$ belongs to the same $V'\in\pi^y_{\mathcal{P}_y,\sigma_y}$, then there clearly exists an $u,w$-path inside $V$ by the induction hypothesis. Therefore, let $u\in V_i$ and $v\in V_j$ for distinct $V_i,V_j\in\pi^y_{\mathcal{P}_y,\sigma_y}$. Let $u'\in V_i\cap\beta(y)$ and $v'\in V_j\cap\beta(y)$. By the definition of split, there exists a sequence $j_1,\ldots,j_k$, with $j_1 = \mathcal{P}_y(u')$, $j_k = \mathcal{P}_y(v')$, and such that $P_y^{-1}(j_\alpha)\cap P_z^{-1}(j_{\alpha+1}) \not=\emptyset$ and $P_z^{-1}(j_\beta)\cap P_y^{-1}(j_{\beta+1})\not=\emptyset$. As each opened part is connected, there exists a sequence of subpaths using the vertices in non-empty intersections of consecutive parts $j_1,j_2,\ldots,j_k$. The same argumentation holds also for $u,w\in V^z$ and $u\in V^y$ and $w\in V^z$. This finishes the proof. 
	\end{proof}
	
	Next, we show the opposite implication, that is, if there is a partial solution in a sub-tree rooted in some node $x$, then the table for the corresponding signature always stores \texttt{true}.
	
	\begin{lemma}\label{thm:ECP:XP:tw:falseVal}
		For every node $x$ and every pair $(\mathcal{P},\sigma)$, if there does exist a partial solution~$\pi^x_{\mathcal{P},\sigma}$ compatible with $(\mathcal{P},\sigma)$, then $\DP_x[ \mathcal{P}, \sigma ] = \mathtt{true}$.
	\end{lemma}
	\begin{proof}
		We again prove the lemma by a bottom-up induction on the tree decomposition of~$G$. 
		We start by showing that the lemma holds for leaves. 
		Let $\pi^x_{\mathcal{P},\sigma}$ be a partial solution compatible with $(\mathcal{P},\sigma)$. Since $V^x=\emptyset$, also $\pi^x_{\mathcal{P},\sigma}$ is an empty partition. Therefore, $\mathcal{P}^{-1}(i) = \emptyset$ and $\sigma(i)=0$ for every $i\in[\tw+1]$. However, in this case, we set $\DP_x[\mathcal{P},\sigma] = \mathtt{true}$ by the definition of the computation. Hence, for leaves, the lemma clearly holds. Let us now suppose that the lemma holds for all children of node $x$.
		
		Let $x$ be an introduce node with exactly one child $y$, let $\{v\} = \beta(x)\setminus\beta(y)$, and let $\pi^x_{\mathcal{P},\sigma}$ be a partial solution in $x$. First, suppose that $v$ is in a singleton part in $\pi^x_{\mathcal{P},\sigma}$. If we remove the singleton part the vertex $v$ is part of, we obtain a solution partition $\pi^y_{\mathcal{P}',\sigma}$, where $\mathcal{P} = \mathcal{P}' \Join (v \mapsto i)$ for some $i$ such that $\mathcal{P}'^{-1}(i) = \emptyset$. By the induction hypothesis, $\DP_y[\mathcal{P'},\sigma]$ is $\mathtt{true}$ and therefore, also $\DP_x[\mathcal{P},\sigma]$ is $\mathtt{true}$ by the definition of the computation. Now, let $v$ be in a part $V_i$ which contains more vertices. Clearly, there is at least one other bag vertex, as otherwise, $V_i$ would not be connected. We remove $v$ from $V_i$ and add each connected component that arise from removal of $v$ from $V_i$ and replace with them the part $V_i$ in $\pi^x_{\mathcal{P},\sigma}$ to obtain a partial solution $\pi^y_{\mathcal{P}',\sigma'}$. By induction hypothesis, $\DP_y[\mathcal{P}',\sigma']$ is $\mathtt{true}$ and consequently, also $\DP_x[\mathcal{P},\sigma]$ is necessarily set to $\mathtt{true}$.
		
		Next, let $x$ be a forget node with a single child $y$, let $v$ be the forgotten vertex, and let~$\pi^x_{\mathcal{P},\sigma}$ be a partial solution. Moreover, let $V_i$ be a part the vertex $v$ is member of. First, suppose that $V_i$ is a closed part. It follows that in $y$, there exists a $\mathcal{P}' = \mathcal{P} \Join (v \mapsto i)$ and $\sigma' = \sigma \Join (i \mapsto (|V_i| - 1))$ such that $\pi^x_{\mathcal{P},\sigma}$ is compatible with $(\mathcal{P}',\sigma')$. By induction hypothesis, $\DP_y[\mathcal{P}',\sigma']$ is set to $\mathtt{true}$. However, by the definition of the computation, also $\DP_x[\mathcal{P},\sigma]$ is set to $\mathtt{true}$. Next, let $V_i$ be an opened part and, without loss of generality, let $V_i\cap\beta(x) = \mathcal{P}^{-1}(i)$, where $\mathcal{P}^{-1}(i) \not= \emptyset$. Consequently, $\pi^x_{\mathcal{P},\sigma}$ is also a partial solution for a signature $(\mathcal{P}' = \mathcal{P} \Join (v \mapsto i), \sigma' = \sigma \Join (i \mapsto (\sigma(i) - 1)))$. By the induction hypothesis, $\DP_y[ \mathcal{P}', \sigma' ]$ is $\mathtt{true}$, and therefore, by the definition of computation, also $\DP_x[\mathcal{P},\sigma]$ is $\mathtt{true}$. This finishes the correctness for forget nodes.
		
		Finally, let $x$ be a join node with exactly two children $y$ and $z$ such that $\beta(x) = \beta(y) = \beta(z)$, and let~$\pi^x_{\mathcal{P},\sigma}$ be a partial solution. 
		We create $\pi^y_{\mathcal{P}_y,\sigma_y}$ by removing all past vertices appearing only in $V^z$, and similarly we create $\pi^z_{\mathcal{P}_z,\sigma_z}$. As no closed part in $\pi^x_{\mathcal{P},\sigma}$ may consists of vertices from both $V^y$ and $V^z$, then all closed parts in $\pi^y_{\mathcal{P}_y,\sigma_y}$ and $\pi^z_{\mathcal{P}_z,\sigma_z}$ are of correct size and are connected. Moreover, for each component opened part that is now disconnected in $\pi^y_{\mathcal{P}_y,\sigma_y}$ or $\pi^z_{\mathcal{P}_z,\sigma_z}$, we create its own opened part. By the induction hypothesis, both $\DP_y[\mathcal{P}_y,\sigma_y]$ and $\DP_z[\mathcal{P}_z,\sigma_z]$ are $\mathtt{true}$ and, by the definition of the computation, also $\DP_x[\mathcal{P},\sigma]$ is necessarily set to $\mathtt{true}$. This finishes the proof of the lemma.
	\end{proof}
	
	By \Cref{thm:ECP:XP:tw:trueVal,thm:ECP:XP:tw:falseVal}, we obtain that the dynamic programming table for each node and each signature stores $\mathtt{true}$ if and only if there is a corresponding partial solution. Therefore, the computation, as was described, is correct. By \Cref{cl:TW:DP:tableSize}, the size of each dynamic programming table is $n^\Oh{\tw}$, and each table can be filled in the same time (see \Cref{lem:leafRunTime,lem:introduceRunTime,lem:forgetRunTime,lem:joinRunTime}). Therefore, the algorithm runs in $\Oh{n\cdot \tw}\cdot n^\Oh{\tw}\cdot n^\Oh{\tw}$ time, which is indeed in \XP.
\end{proof}

Observe that if the size of every part is bounded by a parameter $\varsigma$, the size of each dynamic programming table is $\tw^\Oh{\tw}\cdot\varsigma^\Oh{\tw}$ and we need the same time to compute each cell. Therefore, the algorithm also shows the following tractability result.

\begin{corollary}
	The \ECP problem is fixed-parameter tractable parameterised by the tree-width $\tw$ and the size of a large part $\varsigma = \lceil n/p \rceil$ combined.%
\end{corollary}

In other words, the \ECP problem becomes tractable if the tree-width is bounded and the number of parts is large.

So far, we investigated the complexity of the problem mostly with respect to structural parameters that are bounded for sparse graphs. Now, we turn our attention to parameters that our bounded for dense graphs. Note that such parameters are indeed interesting for the problem, as the problem becomes polynomial-time solvable on cliques. We were not able to find this result in the literature, and, therefore, we present it in its entirety.
	
\begin{observation}
	\label{thm:P:clique}
	The \ECP problem can be solved in linear time if the graph $G$ is a clique.
\end{observation}
\begin{proof}
	First, we determine the number of parts of size $\lceil n/p \rceil$ as $\ell = (n\bmod p)$, and the number of smaller parts of size $\lfloor n/p \rfloor$ as $s = p - \ell$. Now, we arbitrarily assign vertices to $p$ parts such that the first $\ell$ parts contain $\lceil n/p \rceil$ vertices and the remaining $s$ parts contain exactly $\lfloor n/p \rfloor$ vertices. This, in fact, creates an equitable partition. Moreover, every partition is connected, since each pair of vertices is connected by an edge in $G$.
\end{proof}
	
Following the usual approach of distance from triviality~\cite{AgrawalR2022,GuoHN2004}, we study the problem of our interest with respect to the distance to clique. We obtain the following tractability result.

\begin{theorem}
	\label{thm:ECP:FPT:distClique}
	The \ECP problem is fixed-parameter tractable when parameterised by the distance to clique~$k$.
\end{theorem}
\begin{proof}
	We partition $V$ into the clique $K$ and the modulator $M$ where $|M|=k$.
	We label each vertex in $K$ by its set of neighbours in $M$.
	These labels partition $K$ into no more than~$2^k$ groups, each group with its own label.
	Note that each group consists of mutual siblings.
	We design a branching procedure where each decision is either forced or explored completely.
	Hence, whenever any part grows to more than $\lceil n/p \rceil$ vertices, we discard that branch.
	
	To create an equitable connected partition, we need to tackle the parts that overlap $M$.
	First, we branch into every possibility of partitioning the modulator $M$.
	This gives us Bell number of branches (at most $(\frac{0.792k}{\ln(k+1)})^k$ \cite{BerendT2000}), which we upper bound by $k^k$.
	Next, we branch into at most $2^k$ decisions on whether each of the modulator parts has size $\lceil n/p \rceil$ or $\lfloor n/p \rfloor$.
	In each of these branches we need to ensure that every modulator part is connected -- not only to other vertices within $M$ but also to the rest of the part in $K$ that will be assigned to the part later.
	We branch for each vertex $u \in M$ whether it is connected to its part via modulator edges or through some vertex of $K$.
	As we noted, each group of $K$ consists of siblings, so membership of particular vertices of $K$ to the parts may be arbitrarily permuted.
	Moreover, it is clear that for a fixed group, if a part in $M$ also contains a vertex of the group, then it does not need to contain more than one such vertex to ensure connectedness.
	Hence, this step creates at most $(2^k+1)^k$ branches.
	We discard branches that have too few vertices in groups of $K$ to satisfy the demands of our connection decisions.
	Finally, if the modulator parts are still disconnected, then discard the branch.
	
	We now need to assign the remaining vertices of $K$.
	First to fill up the modulator parts to the predetermined size and then to create parts that lie entirely within $K$.
	We follow with an argument that shows that the filling up can be done greedily:
	Assume that $S$ is the partition we seek.
	Then for any part $P \in S$ that contains vertices of $K$ we can find a spanning tree $T$ of $P$ and root it at an arbitrary vertex of $P \cap K$.
	Every vertex of $P \cap M$ has in $T$ a parent -- which can be chosen in the procedure that ensures connectivity.
	In a branch with such choices for all parts of $S$ we know that if the part intersects $K$ then it has already some vertex of $K$ assigned to it.
	Hence, by assigning vertices of $K$ to the modulator parts greedily we will not skip the solution.
	The rest -- creating parts from the vertices of the remaining clique -- is done using \Cref{thm:P:clique}.
	
	As each part that goes over all vertices can be done in linear time, the final time complexity of the algorithm is $k^k \cdot 2^k \cdot (2^k+1)^k \cdot \Oh{n+m} = 2^\Oh{k^2}\cdot (n+m)$, which is indeed \FPT.
\end{proof}

Using similar arguments as in \Cref{thm:P:clique}, we can prove polynomial-time solvability by providing an algorithm for a more general class of graphs than cliques. Namely, we provide a tractable algorithm for co-graphs.

\begin{theorem}
	\label{thm:P:cograph}
	The \ECP problem can be solved in polynomial time if the graph $G$ is a co-graph.
\end{theorem}
\begin{proof}
	For $\lceil \frac np \rceil \le 2$ one can solve any instance using maximum matching, so assume $\frac np > 2$.
	First, construct a binary co-tree of $G$ where nodes represent that either we took disjoint union of two sub-graphs or we get a complete-join of a disjoint union.
	We perform a standard bottom-up dynamic programming algorithm.
	The goal is to compute for each co-tree node~$N$ a table of size $n+1$ that contains on index $i \in \{0,\dots,n\}$ value $N_i = \texttt{True}$ if and only if the sub-graph under $N$ can be partitioned into parts of correct sizes ($\lfloor\frac np\rfloor, \lceil\frac np\rceil$) and a single part of size exactly $i$, otherwise $N_i$ is $\texttt{False}$.
	
	Before running the dynamic programming we precompute for every pair of values $k,\ell \in \{0,\dots,n\}$ a value $K(k,\ell) \in \{\texttt{True},\texttt{False}\}$ that expresses whether a complete bipartite graph $K_{k,\ell}$ has an \ECPshort solution.
	This can be done via a simple dynamic programming on increasing values of $k$ and $\ell$ in time $\mathcal O(n^3)$.
	
	The dynamic programming algorithm over co-tree goes as follows.
	As a base step, we have a single vertex; there we set $N_1 = \texttt{True}$ and the rest to $\texttt{False}$.
	An inductive step, take node $N$ which has children $A$ and $B$ that have their tables already computed.
	If $N$ is a disjoint union then set $N_i = \bigvee_{j \in \{0,\dots,i\}} (A_j \land B_{i-j})$.
	If $N$ is a complete-join node,
	set $N_i = \bigvee_{j,k,\ell \in \{0,\dots,i\}} \big(A_{j+k} \land B_{i-j+\ell} \land K(k,\ell)\big)$, where the value for infeasible indices is $\texttt{False}$.
	
	If we have $N_0 = \texttt{True}$ for the root node $N$ of the co-tree, then there exists a \ECPshort solution, because the algorithm decreases the number of non-partitioned vertices only via covering them by $K(k,\ell)$, which was precomputed to have a proper partitioning.
	In the other direction, assume we have a \texttt{Yes} \ECPshort instance that has a solution $S$.
	Take an arbitrary part $P$ of $S$.
	Identify the edge within $P$ that is introduced in the node $N_P$ that is closest to the root of the co-tree.
	As $N_P$ introduced an edge it is a complete-join node and $P$ spans both sub-graphs that are joined by it.
	Let $\mathcal N_N$ be the set of all parts that are assigned to $N$ using the described procedure.
	Let $\mathcal N^\downarrow_N$ be union of all sets $\mathcal N_A$ of all descendants $A$ of $N$ (including $N$).
	By definition, for any node $N$ and its children $A$ and $B$ we have $\mathcal N^\downarrow_N = \mathcal N^\downarrow_A \cup \mathcal N^\downarrow_B \cup \mathcal N_N$.
	We claim that every node $N$ has a $\texttt{True}$ entry that corresponds to $\mathcal N^\downarrow_N$.
	We set $L_1$ to $\texttt{True}$ for any co-tree leaf node $L$.
	Within the solution $S$ the set $\mathcal N_L$ is empty because $n/p > 2$.
	For induction, assume $\mathcal N^\downarrow_A$ and $\mathcal N^\downarrow_B$ have entries $A_{j'}$ and $B_{i'}$ set to $\texttt{True}$ where $j'$ and $i'$ represent the number of vertices that are not covered in $\mathcal N^\downarrow_A$ and $\mathcal N^\downarrow_B$, respectively.
	Let $k$ be the number of vertices of $A$ in (the union of sets within) $\mathcal N_N$ and let $\ell$ be the number of vertices of $B$ in $\mathcal N_N$.
	The number of vertices in $N^\downarrow_N$ that are not covered is $i=i'+j'-k-\ell$.
	Our algorithm sets $N_i$ to $\texttt{True}$ because for $j=j'-k$ and $i=i'+j'-k-\ell$ we have $A_{j+k} \land B_{i-j+\ell} \land K(k,\ell) = A_{j'} \land B_{i'} \land K(k,\ell) = \texttt{True}$.
	
	The algorithm goes through all nodes of the co-tree in the bottom-up order and performs for each node $N$ computation of $N_i$ for all $i \in \{0,\dots,n\}$ that takes no more than $\mathcal O(n^3)$.
	So the total time complexity, including the precomputation, is $\mathcal O(n^4)$.
\end{proof}

Next structural parameter we study is the neighbourhood diversity, which is generalisation of the famous vertex cover number that, in contrast, allows for large cliques to be present in~$G$. Later on, we will also provide a fixed-parameter tractable algorithm for a more general parameter called modular-width; however, the algorithm for neighbourhood diversity will serve as a building block for the later algorithm, and therefore we find it useful to present the algorithm in its entirety.

\begin{theorem}\label{thm:ECP:FPT:nd}
	The \ECP problem is fixed-parameter tractable parameterised by the neighbourhood diversity $\nd(G)$.
\end{theorem}
\begin{proof}
	We first observe that each connected sub-graph of~$G$ ``induces'' a connected graph of the type-graph of~$G$.
	More precisely, each connected sub-graph of~$G$ is composed of vertices that belong to types of~$G$ that induce a connected sub-graph of the type-graph.
	Therefore, the solution is composed of various realisations of connected sub-graphs of the type-graph of~$G$.
	Note that there are at most $2^{\nd(G)}$ connected sub-graphs of the type-graph of~$G$.
	We will resolve this task using an ILP using integral variables $x_H^t$ for a type $t$ and a connected sub-graph~$H$ of the type-graph of~$G$.
	Furthermore, we have additional variables $x_H$.
	That is, the total number of variables is (upper-bounded by) $\nd(G) \cdot 2^{\nd(G)} + 2^{\nd(G)}$.
	The meaning of a variable $x_H^t$ is ``how many vertices of type~$t$ we use in a realisation of~$H$''.
	The meaning of a variable $x_H$ is ``how many realisations of~$H$ there are in the solution we find''.
	We write $t \in H$ for a type that belongs to~$H$ (a connected sub-graph of the type-graph).
	Let $\sigma$ denote the lower-bound on the size of parts of a solution, that is, $\sigma = \lfloor n / k \rfloor$.
	In order for this to hold we add the following set of constraints (here, $\xi_G = 0$ if $n = k \cdot \sigma$ and $\xi_G = 1$, otherwise):
	\begin{align}
		\sigma x_H \le \sum_{t \in H} x_H^t &\le (\sigma + \xi_G) x_H && \forall H \label{eq:ECP:FPT:nd:realised} \\
		x_H &\le x_H^t 				 && \forall H \, \forall t \in H \label{eq:ECP:FPT:nd:realised:lowerbound} \\
		\sum_{H} x_H^t &= n_t 	&& \forall t \in T(G) \label{eq:ECP:FPT:nd:types}	\\
		0 \le x_H^t \le n_t &,\, x_H^t \in \mathbb{Z} 	&& \forall H \, \forall t \in H \label{eq:ECP:FPT:nd:bounds:Ht} \\
		0 \le x_H &,\, x_H^t \in \mathbb{Z} 	&& \forall H \label{eq:ECP:FPT:nd:bounds:H}
	\end{align}
	That is, \eqref{eq:ECP:FPT:nd:types} ensures that we place each vertex to some sub-graph $H$.
	The set of conditions \eqref{eq:ECP:FPT:nd:realised} ensures that the total number of vertices assigned to the pattern~$H$ is divisible into parts of sizes $\sigma$ or $\sigma+1$.
	The set of conditions \eqref{eq:ECP:FPT:nd:realised:lowerbound} ensures that each type that participates in a realisation of~$H$ contains at least~$x_H$ vertices, that is, we can assume that each realisation contains at least one vertex of each of its types.
	It is not difficult to verify that any solution to the \ECP problem fulfils \eqref{eq:ECP:FPT:nd:realised}--\eqref{eq:ECP:FPT:nd:bounds:H}.
	
	In the opposite direction, suppose that we have an integral solution $x$ satisfying \eqref{eq:ECP:FPT:nd:realised}--\eqref{eq:ECP:FPT:nd:bounds:H}.
	Let $\cal{H}$ be the collection of graphs~$H$ with multiplicities corresponding to~$x$, that is, a graph~$H$ belongs to~$\mathcal{H}$ exactly $x_H$-times.
	First, we observe that $|\mathcal{H}| = k$.
	To see this note that
	\[
	|G|
	=
	\sum_{t \in H} n_t
	=
	\sum_{t \in H}\sum_{H \in \mathcal{H}} x_H^t
	=
	\sum_{H \in \mathcal{H}}\sum_{t \in H} x_H^t
	\ge
	\sum_{H \in \mathcal{H}} \sigma \cdot 1
	\ge
	\sigma \sum_H x_H
	=
	\sigma |\mathcal{H}|
	\,.
	\]
	Similarly, we have $|G| \le (\sigma + \xi_G) |\mathcal{H}|$ and the claim follows.
	Now, we find a realisation for every $H \in \mathcal{H}$.
	We know that there are $x_H^t \ge \sigma x_H$ vertices allocated to~$H$.
	We assign them to the copies of $H$ in $\mathcal{H}$ as follows.
	First, from each type $t \in H$ we assign one vertex to each copy of~$H$ (note that this is possible due to~\eqref{eq:ECP:FPT:nd:realised:lowerbound}).
	We assign the rest of the vertices greedily, so that there are $\sigma$ vertices assigned to each copy of~$H$; then, we assign the leftover vertices (note that there are at most~$x_H$ of them in total) to the different copies of $H$.
	In this way, we have assigned all vertices and gave a realisation of~$\mathcal{H}$.
	
	As was stated before, the integer linear program has only parameter-many variables. Hence, we can use the algorithm of Lenstra~\cite{Lenstra1983} to solve it in \FPT time.
\end{proof}

With the algorithm from the proof of \Cref{thm:ECP:FPT:nd} in hand, we are ready to derive the result also for the modular-width.

\begin{theorem}
	\label{thm:ECP:FPT:mw}
	The \ECP problem is fixed-parameter tractable parameterised by the modular-width $\mw(G)$.
\end{theorem}
\begin{proof}
	Clearly, the leaf-nodes of the modular decomposition of~$G$ have bounded neighbourhood diversity.
	For the graph of a leaf-node, we employ the following ILP:
	\begin{align}
		\sigma x_H \le \sum_{t \in H} x_H^t &\le (\sigma + \xi_G) x_H && \forall H \label{eq:ECP:FPT:mw:realised} \\
		x_H &\le x_H^t 				 && \forall H \, \forall t \in H \label{eq:ECP:FPT:mw:realised:lowerbound} \\
		\sum_{H} x_H^t &\le n_t 	&& \forall t \in T(G) \label{eq:ECP:FPT:mw:types}	\\
		0 \le x_H^t \le n_t &,\, x_H^t \in \mathbb{Z} 	&& \forall H \, \forall t \in H \label{eq:ECP:FPT:mw:bounds:Ht} \\
		0 \le x_H &,\, x_H^t \in \mathbb{Z} 	&& \forall H \label{eq:ECP:FPT:mw:bounds:H}
	\end{align}
	Note the difference between \eqref{eq:ECP:FPT:mw:types} and \eqref{eq:ECP:FPT:nd:types}; that is, this time we do not insist on assigning all vertices and can have some leftover vertices.
	We add an objective function
	\[
	\max \sum_{H}\sum_{t \in H} x_H^t \,,
	\]
	that is, we want to cover as many vertices as possible already in the corresponding leaf-node.
	Next, we observe that based on the solution of the above ILP, we can replace the leaf-node with a graph of neighbourhood diversity~2.
	In order to do so, we claim that if we replace the graph represented by the leaf-node by a disjoint union of a clique of size $\sum_{H}\sum_{t \in H} x_H^t$ and an independent set of size $\sum_{H} (n_t - \sum_{t \in H} x_H^t)$, then we do not change the answer to the \ECP problem.
	That is, the answer to the original graph was \Yes{} if and only if the answer is \Yes{} after we alter the leaf-node.
	The algorithm for modular-width then follows by a repeated application of the above ILP.
\end{proof}

As the last result of this section, we give an \XP algorithm for another structural parameter called distance to cluster graph. The algorithm, in its core, is based on the same ideas as our \FPT algorithm for distance to clique. Nevertheless, the number of types of vertices is no longer bounded only by a function of a parameter, and to partition the vertices that are not in the neighbourhood of the modulator vertices, we need to employ dynamic programming.

\begin{theorem}
	\label{thm:ECP:XP:dcg}
	The \ECP problem is \XP parameterised by the distance to cluster graph $\operatorname{dcg}(G)$.
\end{theorem}
\begin{proof}
	We partition $V$ into the cluster graph vertices $R$ and the modulator $M$ where $|M|=k$.
	Recall that the cluster graph $R$ is a disjoint union of complete graphs.
	We design a branching procedure with a dynamic programming subprocedure.
	Whenever any part grows to more than $\lceil n/p \rceil$ vertices, we discard the branch.
	
	This proof bears similarity to the proof of \Cref{thm:ECP:FPT:distClique} with two main differences.
	First, we cannot ensure that the number of types (based just on neighbourhood to $M$) is bounded by the parameter so the final complexity becomes \XP.
	Second, to partition the remaining vertices of $R$ we employ a dynamic programming approach.
	
	We begin by going over every possibility of partitioning the modulator $M$, which creates no more than $k^k$ branches.
	We defer correctness of the following step to the end of this proof.
	For each $u \in M$ we add to its part up to two vertices $v_1,v_2 \in R$ such that $\{u,v_1\} \in E$ and $\{v_1,v_2\} \in E$, resulting in no more than $n^{2k}$ branches (just for counting, adding a single vertex is as if $v_1=v_2$ and adding none as $u=v_1=v_2$).
	We now check whether all modulator parts are connected and discard the branch if not.
	
	Having partially decided modulator parts we now perform a dynamic programming that reveals whether it is possible to partition the cluster graph correctly.
	Let $K^1,\dots,K^r$ be the cliques of $R$.
	We have a table $\DP$ with boolean entries $\DP[i,j_1,j_2,\dots,j_k]$ where $i$ ($0 \le i \le r$) is the index of the last processed clique $K^i$ ($i=0$ if none is processed) and $j_\ell$ ($1 \le j_\ell \le n$) gives how many vertices are assigned to the $\ell$-th modulator partition.
	Of course there may be fewer than $k$ modulator partitions, we assume $k$ for simpler indexing.
	As $r \le n$ the size of $\DP$ is at most $(n+1)^{k+1}$.
	The boolean entry $\DP[i,j_1,j_2,\dots,j_k]$ represents whether it is possible to partition $K_1,\dots,K_i$ in addition to the partial modulator partition (result of the branching) such that the $i'$-th modulator partition has size $j_{i'}$.
	To compute the $\DP$ table we start with $i=0$.
	We assign \texttt{True} to the only entry that has $j_{i'}$ equal to the current size of the $i'$-th modulator part for all $1 \le i' \le k$; the rest is set to \texttt{False}.
	Further entries are computed by gradually increasing $i$.
	To compute $\DP[i,j_1,j_2,\dots,j_k]$ we split each $j_{i'}$ into a sum $j_{i'} = a_{i'} + b_{i'}$ of two non-negative numbers.
	There are exactly $j_{i'}+1 \le n+1$ ways to do this split.
	The total number of options over all $i'$ is no more than $(n+1)^k$.
	To set $\DP[i,j_1,j_2,\dots,j_k]:=\texttt{True}$ at least one split has to satisfy both of the following conditions.
	\begin{enumerate}[1)]
		\item $\DP[i-1,a_1,a_2,\dots,a_k]$ is $\texttt{True}$ and
		\item\label{case2} the vertices of $K^i$ that are yet not in any part can be partitioned into parts contained entirely in $K^i$ and such that for all $1 \le i' \le k$ we can connect $b_{i'}$ vertices to the $i'$-th modulator part.
	\end{enumerate}
	We need to ensure that each modulator that is extended is connected to the added vertices.
	As $K^i$ is a clique we need at least one edge between the modulator parts and its new vertices in $K^i$.
	This can be checked using any polynomial-time maximum matching algorithm in the following way:
	We do not check modulator parts that already have vertices in $K^i$ (from the branching procedure).
	For the rest, build a bipartite graph $B$ representing the modulator parts and their neighbours in $K^i$.
	If the maximum matching algorithm finds a solution that matches all modulator parts then the matched vertices of $K^i$ ensure the connectivity.
	
	After evaluating the entire $\DP$ table, we check all entries $\DP[r,j_1,\dots,j_k]$ where $j_{i'}$ has correct size ($\lceil n/p \rceil$ or $\lfloor n/p \rfloor$) for all $1 \le i' \le k$.
	Note that the conditions ensure that parts that are contained entirely within the cliques have correct size.
	Moreover, we keep sizes of modular parts until the very end and only the entries with correct sizes are checked.
	Therefore, we have \YesI if and only if any of the checked entries for $i=r$ is \texttt{True}.
	
	Having an intuition about the algorithm, we now return back to the argument about adding vertices to modulator parts during the branching procedure.
	Assume that $S$ is the partition we seek.
	We prove that the described branching suffices by a contradiction.
	Pick a part $P$ in $S$ and assume that $P$ was discarded during our branching procedure.
	There must be two vertices $u,v \in P \cap M$ such that the shortest path $X$ between them within $P$ does not contain other vertices of $P \cap M$ and contains more than two vertices of $P \cap R$.
	Path vertices $X \cap R$ are contained in a single clique $R$ and could be shortcut to make a shorter path between $u$ and $v$, a contradiction.
	
	The time complexity of the branching process is $k^k \cdot n^{2k} \cdot \Oh{n+m}$ and of the dynamic programming is $(n+1)^{2(k+1)}\cdot (n+m)^{\Oh{1}}$.
	We are getting an \XP algorithm that runs in $f(k) \cdot n^{\Oh{k}}$ for a computable function $f$.
\end{proof}

\section{Hardness Results}\label{sec:hardness}

In this section, we complement our algorithmic upper-bounds from the previous section with matching hardness lower-bounds. The results from this section clearly show that no \XP algorithm introduced in this paper can be improved to a fixed-parameter tractable one, or pushed to a more general parameter.

First, we observe that \ECPshort is \Wh with respect to the feedback-edge set number $\fes$ of $G$.
This negatively resolves the question from the introduction of our paper.
In fact, the actual statement shows an even stronger intractability result.

\begin{theorem}
	\label{ECP:Wh:fes}
	The \ECP problem is \Wh with respect to the path-width $\pw(G)$, the feedback-edge set number $\fes(G)$, and the number of parts~$p$ combined.
\end{theorem}
\begin{proof}
	Enciso et al.~\cite{EncisoFGKRS2009} show in their reduction that the \ECPshort problem is \Wh parameterised by the path-width, the feedback-vertex set, and the number of parts combined.
	We observe, that the same reduction also shows \Whness for the path-width, the feedback-edge set, and the number of parts combined.
	In short, vertices they remove to achieve bounded feedback-vertex set have only constant number of edges connecting them with the rest of the construction.
	
	To be more precise, the proof of Enciso et al.~\cite[Theorem 1]{EncisoFGKRS2009} performs a reduction from the \MCC problem.
	Let us briefly give you an idea of their construction.
	The first key idea is to construct a set of vertices, called \emph{(roots of) anchors}, that have many pendant leaves.
	When the number of leaves is large enough, the solution is forced to contain exactly one part per anchor.
	The construction then connects the anchors using caterpillars (paths with pendant leaves) in a particular way.
	Observe that the caterpillars restrict sizes of the parts -- with each vertex of the path the part must also include all its pendant vertices.
	
	For each of the $k$ colours, they build a vertex selection gadget -- a cycle with $2(k-1)$ anchors joined with caterpillars.
	There are two consecutive anchors representing connection to other colours.
	The joining caterpillars have several vertices with big-degree, and in between them vertices of smaller degree.
	As the big degrees are way bigger than anything else the part can include, all the parts in the vertex selection gadget are forced to be offset by the same amount of big-degree vertices along the direction of the cycle -- this represents choice of a vertex of the given colour.
	They are all forced into the same area where they select small-degree vertices -- these represent what vertices of other colours the selected vertex is compatible with (the adjacent ones).
	
	For every colour, we have a vertex selection gadget that has parts representing every other colour.
	So now, for each pair of colours $i,j$ we build a caterpillar that joins the first vertex that represents colour $j$ in the vertex selection gadget for colour $i$ with the first vertex that represents colours $i$ in the vertex selection gadget for colour $j$; and similarly connects the second vertices.
	This creates a cycle with four anchors.
	Similarly to the vertex selection gadgets that ``synchronised'' using big-degree vertices, these smaller cycles now synchronise using small-degree vertices, allowing only offsets that represent existing edges.
	These offsets ensure that edges between the selected vertices exist for each pair of colours -- representing a multicoloured clique.
	
	Observe that there are $2k(k-1)$ anchors.
	Their removal yields a graph that has only caterpillars and independent vertices.
	To finish, we see that the number of parts $p$ is equal to the number of anchors,
	path-width $\pw \le 2k(k-1)+1$, and $\fes \le 4k(k-1)$ as to remove all cycles, it suffices to remove the edges of caterpillars in vertex selection gadgets that are adjacent to the anchors.
\end{proof}

Note that the result from \Cref{ECP:Wh:fes} can be, following the same arguments as used by~\cite{EncisoFGKRS2009}, strengthened to show the same result for planar graphs.
	
\begin{corollary}
	The \ECP problem is \Wh when parameterised by the path-width, the feedback-edge set, and the number of parts combined, even if $G$ is a planar graph.
\end{corollary}

Next, we show that the \pvcn parameter from \Cref{thm:ECP:FPT:3pvc} cannot be relaxed any more while the problem is kept tractable. Our reduction is even more general and comes from the \UBP problem, which is defined as follows.
	
\defProblemQuestion{\UBP}
	{A number of bins~$k$, a capacity of a single bin~$b$, and a multi-set of integers $A=\{a_1,\ldots,a_n\}$ such that $\sum_{a\in A} a = bk$.}
	{Is there a surjective mapping $\alpha\colon A \to [k]$ such that for every~$i\in[k]$ we have $\sum_{a\in\alpha^{-1}(i)} a = b$?\vspace{0.2em}}
	
The \UBP problem is well-known to be \Wh when parameterised by the number of bins~$k$ and not solvable in $f(k)\cdot n^{o(k/\log k)}$ time for any computable function~$f$, even if all numbers are given in unary~\cite{JansenKMS2013}.

\begin{theorem}
	\label{thm:ECP:Wh:4pvc}
	For every graph family $\mathcal{G}$ such that it contains at least one connected graph~$G$ with $s$ vertices for every $s\in\mathbb{N}$, the \ECP problem is \Wh parameterised by the distance to $\mathcal{G}$ referred to as $\operatorname{dist}_\mathcal{G}(G)$ and the number of parts~$p$ combined and, unless ETH fails, there is no algorithm running in $f(\ell)\cdot n^{o(\ell/\log{\ell})}$ time for any computable function~$f$, where $\ell = p + \operatorname{dist}_{\mathcal{G}}(G)$.
\end{theorem}
\begin{proof}
	Let $\mathcal{I}=(A,k,b)$ be an instance of the \UBP problem. We construct an equivalent instance $\mathcal{J}=(G,p)$ of the \ECP problem as follows (see \Cref{fig:UBP_reduction} for an overview of the construction). For the sake of exposition, we assume that $\mathcal{G}$ is a family containing all disjoint unions of stars; later we show how to tweak the construction to work with any~$\mathcal{G}$ satisfying the conditions from the theorem statement.
	
	For every number $a_i \in A$, we create a single \emph{item-gadget} $S_i$ which is a star with $a_i$ vertices. Every $S_i$ will be connected with the rest of the graph~$G$ only via the star centre $c_i$; we call this special vertex a \emph{hub}. Next, we create $k$ \emph{bin-gadgets} $B_1,\ldots,B_k$. Each of these gadgets consists of a single vertex. Slightly abusing the notation, we call this vertex also $B_i$, $i\in[k]$. As the last step of the construction, we add an edge connecting every bin-gadget with every item-gadget and set $p=k$.
	
	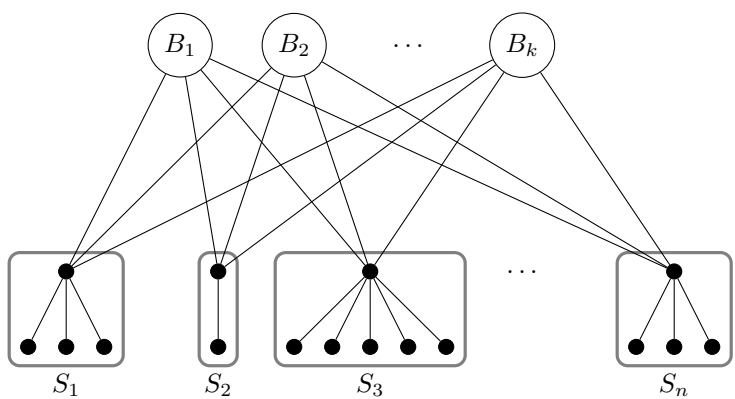
\begin{figure}[tb!]
		\centering
		\begin{tikzpicture}
			\node[draw,circle] (B1)  at (-1.5,0) {$B_1$};
			\node[draw,circle] (B2)  at (   0,0) {$B_2$};
			\node              (Bx)  at ( 1.5,0) {$\hdots$};
			\node[draw,circle] (Bk)  at (   3,0) {$B_k$};
			
			\draw[rounded corners,very thick,gray] (5.75, -2.75) rectangle (4.25, -4.25);
			\node at (5,-4.5) {$S_n$};
			\node[draw,fill,circle,inner sep=2pt] (c1) at (   5,-3) {};
			\node[draw,fill,circle,inner sep=2pt] (c11) at (5.5,-4) {};
			\node[draw,fill,circle,inner sep=2pt] (c12) at (  5,-4) {};
			\node[draw,fill,circle,inner sep=2pt] (c13) at (4.5,-4) {};
			\draw (c1) edge (c11) edge (c12) edge (c13);
			
			\node[] (c2) at ( 3,-3) {$\hdots$};

			\draw[rounded corners,very thick,gray] (-0.25, -2.75) rectangle (2.25, -4.25);
			\node at (1,-4.5) {$S_3$};
			\node[draw,fill,circle,inner sep=2pt] (c3) at ( 1,-3) {};
			\foreach[count=\i] \x in {0,0.5,1,1.5,2}{
				\node[draw,fill,circle,inner sep=2pt] (c3\i) at (\x,-4) {};
				\draw (c3) edge (c3\i);
			}
			
			\draw[rounded corners,very thick,gray] (-0.75, -2.75) rectangle (-1.25, -4.25);
			\node at (-1,-4.5) {$S_2$};
			\node[draw,fill,circle,inner sep=2pt] (c4) at (-1,-3) {};
			\node[draw,fill,circle,inner sep=2pt] (c41) at (-1,-4) {};
			\draw (c4) edge (c41);

			\draw[rounded corners,very thick,gray] (-3.75, -2.75) rectangle (-2.25, -4.25);
			\node at (-3,-4.5) {$S_1$};
			\node[draw,fill,circle,inner sep=2pt] (c5) at (  -3,-3) {};
			\node[draw,fill,circle,inner sep=2pt] (c51) at (-2.5,-4) {};
			\node[draw,fill,circle,inner sep=2pt] (c52) at (  -3,-4) {};
			\node[draw,fill,circle,inner sep=2pt] (c53) at (-3.5,-4) {};
			\draw (c5) edge (c51) edge (c52) edge (c53);
			
			\foreach \i in {1,2,k}{
				\foreach \j in {1,3,4,5} {
					\draw (B\i) -- (c\j);
				}
			}
		\end{tikzpicture}
		\caption{An illustration of the construction used to prove \Cref{thm:ECP:Wh:4pvc}.}
		\label{fig:UBP_reduction}
	\end{figure}
	
	For the correctness, let $\mathcal{I}$ be a \YesI and $\alpha$ be a solution mapping. We construct equitable connected partition $\pi=(\pi_1,\ldots,\pi_p)$ for $\mathcal{J}$ as follows. For every $i\in[k]$, we set $\pi_i = \{B_i\} \cup \{V(S_j)\mid a_j\in A \land \alpha(a_j) = i\}$. Since $\alpha$ is a solution for $\mathcal{I}$, the sum of all items assigned to bin $i$ is exactly $b$ and we have that the size of every part $\pi_i$ is exactly $b+1$. Therefore, the partition $\pi$ is clearly equitable. What remains to show is that every $\pi_i$, $i\in[k]$, induces a connected sub-graph of~$G$. It is indeed the case since every partition $\pi_i$ contains at least one bin-gadget which is connected with the centre of every item-gadget whose vertices are members of $\pi_i$. Moreover, by the construction, all vertices of the same item-gadget are always members of the same part.
	
	In the opposite direction, let $\mathcal{J}$ be a \YesI and $\pi=(\pi_1,\ldots,\pi_p)$ be a solution partition. We first prove several properties that every $\pi_i$, $i\in[p]$, has to satisfy.
	
	\begin{claim}\label{lem:bins}
		Every set $\pi_i\in\pi$ contains exactly one bin-gadget $B_j$, $j\in[p]$.
	\end{claim}
		\begin{claimproof}%
			Let there be a set $\pi_i\in\pi$ such that $\pi_i$ contains no bin-gadget. We already argued that the size of every part $\pi_i$ is exactly $b+1$. However, the number of vertices of every item-gadget is at most $b$ and, thus, $\pi_i$ must contains vertices of at least two item-gadgets to reach the bound $b+1$. Since we assumed that $\pi_i$ does not contain any bin-gadget, it follows that $\pi_i$ cannot be connected which contradicts that $\pi$ is a solution. If $\pi_i$ contains two bin-gadgets, then by the pigeonhole principle there is at least one set $\pi_{i'}\in\pi$ without bin-gadget, which is not possible. Therefore, every $\pi_i$ contains exactly one bin-gadget.
		\end{claimproof}
	
	\begin{claim}\label{lem:items}
		Let $S_j$, $j\in[n]$, be an arbitrary item-gadget. If $\pi$ is a solution, then all vertices of $S_j$ are part of the same set $\pi_i \in \pi$.
	\end{claim}
		\begin{claimproof}%
			For the sake of contradiction, let $S_j$ be an item-gadget such that there are two different sets $\pi_i,\pi_{i'}\in\pi$ containing at least one vertex of $S_j$. Without loss of generality, let $\pi_i$ contains the centre $c_j$ of $S_j$ and let $v$ be a vertex of $S_j$ such that $v\in S_j\cap\pi_{i'}$. Then it is easy to see that in graph $G\setminus \pi_i$ the vertex $v$ is an isolated vertex and, therefore, $\pi_{i'}$ cannot induce a connected sub-graph of $G$. This contradicts that $\pi$ is a solution and finishes the proof.
		\end{claimproof}
	
	Without loss of generality, we assume that $B_i\in\pi_i$ for every $i\in[p]$. By \Cref{lem:items}, we have that vertices of single item-gadget are all members of the same set $\pi_i\in\pi$. Therefore, we slightly abuse the notation and use $\pi(S_j)$ to get the identification of the set the vertices of $S_j$ are assigned to. We create a solution for $\mathcal{I}$ as follows. For every item $a_j\in A$, we set ${\alpha(a_j) = \pi(S_j)}$. Due to \Cref{lem:bins}, it follows that for every $i\in[k]$ we have $\sum_{a\colon \alpha(a)=i} a = b$ and, thus, $\alpha$ is indeed a solution for $\mathcal{I}$.
	
	To complete the proof, we recall that we set $p=k$. Additionally, if we remove all bin gadgets from the graph, we obtain a disjoint union of stars. Hence, $\operatorname{dist}_{\mathcal{G}}(G)$ is also $k$. Moreover, the construction can be clearly done in polynomial time.
	To generalise the construction for any $\mathcal{G}$, we just need to replace the item-gadgets. Specifically, for every $a_i$, the corresponding item-gadget is a connected $a_i$-vertex graph $G\in\mathcal{G}$ and the hub $c_i$ is an arbitrary but fixed vertex $v\in G$. The rest of the arguments remain the same.
	
	Finally, assume that there is an algorithm $\mathbb{A}$ solving \ECPshort in time $f(\ell)\cdot n^{o(\ell/\log \ell)}$, where~$f$ is any computable function and $\ell = p + \operatorname{dist}_\mathcal{G}(G)$. Then, given an instance $\mathcal{I}$ of the \textsc{Unary Bin Packing}, we can reduce $\mathcal{I}$ to an equivalent instance $\mathcal{J}$ of \ECPshort, use $\mathbb{A}$ to decide $\mathcal{J}$, and return the same result for $\mathcal{I}$. Since the reduction can be done in polynomial time and $\ell = 2k$, this leads to an algorithm deciding \textsc{Unary Bin Packing} in $f(k)\cdot n^{o(k)}$, which contradicts ETH~\cite{JansenKMS2013}. Therefore, such an algorithm $\mathbb{A}$ cannot exists, unless ETH fails.
\end{proof}

It is not hard to see that every graph $G$ which is a disjoint union of stars has constant $4$-path vertex cover number -- $0$, to be precise. Therefore, by \Cref{thm:ECP:Wh:4pvc} we obtain the desired hardness result.

\begin{corollary}
	The \ECP problem is \Wh parameterised by the $4$-path vertex cover number \pvcn[4] and the number of parts~$p$ combined.
\end{corollary}

Moreover, if the \pvcn[4] is bounded, then so is the tree-depth \td of the graph, because the graph without paths of length more than $k$ has tree-depth at most $k$~\cite[Proposition 6.1]{NesetrilOM2012}.
Removing $\pvcn[4]$ vertices from $G$ could decrease $\td$ of $G$ by at most $\pvcn[4]$ as every removed vertex can remove only one layer in the tree depth decomposition.
	Therefore, $\td \le k + \pvcn[k]$ and we directly obtain.

\begin{corollary}\label{thm:ECP:Wh:td}
	The \ECP problem is \Wh parameterised by the tree-depth $\td(G)$ and the number of parts~$p$ combined.
\end{corollary}

Now, a previous blind spot of our understanding of the \ECP problem's complexity with respect to the structural parameters that are bounded mostly for sparse graphs is the distance to disjoint paths. We again obtain hardness as a direct corollary of \Cref{thm:ECP:Wh:4pvc}.

\begin{corollary}
	\label{thm:Wh:ddp}
	The \ECP problem is \Wh parameterised by the distance to disjoint paths $\operatorname{ddp}(G)$ and the number of parts~$p$ combined.
\end{corollary}

Enciso et al.~\cite{EncisoFGKRS2009} stated (and we formalized in \Cref{thm:ECP:XP:tw}) that there is an \XP algorithm for the \ECP problem parameterised by the tree-width of $G$. A natural question is then whether this algorithm can be improved to solve the problem in the same running-time also with respect to the more general parameter called clique-width. We give a strong evidence that such an algorithm is unlikely in the following theorem.

\begin{theorem}
	\label{thm:ECP:NPh:cliquewidth}
	The \ECP problem is \NPh even if the graph $G$ has clique-width $3$, and is solvable in polynomial-time on graphs of clique-width at most $2$.
\end{theorem}
\begin{proof}
	To show the hardness, we reuse the reduction used to prove \Cref{thm:ECP:Wh:4pvc}. Recall that the construction can be done in polynomial time, thus, the reduction is also a polynomial reduction. What we need to show is that the clique-width of the constructed graph $G$ is constant. We show this by providing an algebraic expression that uses $3$ labels. First, we create a graph $G_1$ containing all bin-gadgets. This can be done by introducing a single vertex and by repeating disjoint union operation. We additionally assume that all vertices in $G_1$ have label~$3$. Next, we create a graph $G_2$ containing all item-gadgets. Every item-gadget is a star which can be constructed using two labels $1$ and $2$. Without loss of generality, we assume that all item-gadgets' centres have label $1$ and all leaves have label~$2$. To complete the construction, we create disjoint union of $G_1$ and $G_2$ and, then, we perform a full join of vertices labelled $1$ and $3$. It is easy to see that the expression indeed leads to a desired graph.
	
	Polynomial-time solvability follows from \Cref{thm:P:cograph} and the fact that co-graphs are exactly the graphs with clique-width at most $2$~\cite{CourcelleO2000}.
\end{proof}

Using similar arguments, we can show \paraNP-hardness also for a more restrictive parameter called shrub-depth~\cite{GajarskyLO2013,GanianHNOMR2012,GanianHNOM2019}.

\begin{theorem}
	\label{thm:NPh:shrubdepth}
	The \ECP problem is \NPh even if the graph~$G$ has shrub-depth $3$.
\end{theorem}
\begin{proof}
	We again use the reduction from \Cref{thm:ECP:Wh:4pvc}. We show an upper-bound on the shrub-depth of this construction by providing a tree-model using $3$ colours and of depth $3$. For every vertex of every bin-gadget, we set the colour to $1$. For every item-gadget centre, we set the colour to $2$, in the remaining vertices are coloured using colour $3$. Now, we create common parent vertex $p_i$ for for all vertices which are part of the item-gadget $X_i$. Finally, we add a root vertex $r$ which is connected with all bin-gadget vertices and with every $p_i$, $i\in[n]$. The edge relation between vertices is than defined as follows:
	\begin{enumerate}
		\item There is an edge between a vertex of colour $2$ and a vertex of colour $3$ if their distance is exactly $2$.
		\item There is an edge between a vertex coloured $1$ and a vertex coloured $2$ if their distance in the tree-model is exactly $3$.
	\end{enumerate}
	For all the remaining combinations, the edge is not present. It is easy to verify that this is indeed a tree-model of the graph from the construction.
\end{proof}

As the last piece of the complexity picture of the \ECP problem, we show that \ECPshort is \Wh with respect to the distance to disjoint cliques. Recall that we give an \XP algorithm for this parameter in \Cref{thm:ECP:XP:dcg}.

\begin{corollary}
	\label{thm:Wh:dcg}
	The \ECP problem is \Wh with respect to the distance to cluster graph $\operatorname{dcg}(G)$ and the number of parts~$p$ combined.
\end{corollary}

\section{Conclusions}\label{sec:conclusions}

We revisit the complexity picture of the \ECP problem with respect to various structural restrictions of the graph. We complement the existing results with algorithmic upper-bounds and corresponding complexity lower-bounds that clearly show that no existing parameterised algorithm can be significantly improved.

Despite that the provided complexity study gives us a clear dichotomy between tractable and intractable cases, there still remain a few blind spots. One of the most interesting is the complexity classification of the \ECP problem with respect to the band-width parameter, which lies between the maximum leaf number and the path-width; however, is incomparable with the feedback-edge set number.

An interesting line of research can target the tightness of our results. For example, we show a clear dichotomy between the tractable and intractable cases of \ECPshort when parameterised by the clique-width. Using similar arguments, we can show the following result for the recently introduced parameter twin-width~\cite{BonnetKTW2022}.

\begin{theorem}
	\label{thm:NPh:twinwidth}
	The \ECP problem is \NPh even if the graph $G$ has twin-width $2$, and is solvable in polynomial-time on graphs of twin-width~$0$.
\end{theorem}
\begin{proof}
	Polynomial-time solvability follows from \Cref{thm:P:cograph} and the fact that a graph has twin-width equal to $0$ if and only if it is a co-graph~\cite{BonnetKTW2022}. For \NPhness, recall the construction from \Cref{thm:ECP:NPh:cliquewidth}. We provide a contracting sequence of maximum red degree equal to $2$. We divide the contractions into three phases. In the first phase, we contract all bin-gadgets to a single vertex $b$. Note that all bin-gadgets are twins and, hence, no red edge appears during these constructions. Next, we independently contract all vertices of every item-gadget except the distinguished vertex $c_i$, $i\in[n]$, to a single vertex $a_i$. Since item-gadgets are cliques, again no red edge can appear, and the resulting sub-graph is a path with two vertices. After these two phases, we end up with a tree, which has twin-width~$2$~\cite{BonnetKTW2022}.
\end{proof}
We conjecture that the problem is polynomial-time solvable also on graphs of twin-width $1$, which, unlike the twin-width $2$ graphs, are additionally known to be recognisable efficiently~\cite{BonnetKRTW2022}. Similarly, we can ask whether the provided \FPT algorithms are optimal under some standard theoretical assumptions, such as the well-known Exponential-Time Hypothesis~\cite{ImpagliazzoP2001}.

Last but not least, the parameterised complexity framework not only gives us formal tools for finer-grained complexity analysis of algorithms for \NPh problems, but, at the same time, equips us also with the necessary formalism for analysis of effective preprocessing, which is widely known as \emph{kernelisation}. A natural follow-up question is then whether the \ECP problem admits a polynomial kernel with respect to any of the studied structural parameters. We conjecture that there is a polynomial kernel with respect to the distance to clique and $3$-path vertex cover number. %

\bibliography{references}
	
\end{document}